\documentclass[12pt,oneside]{amsart}
\usepackage[latin9]{inputenc}
\usepackage{geometry}
\usepackage{mathrsfs}
\usepackage{mathtools}
\usepackage{amstext}
\usepackage{amsthm}
\usepackage{amssymb}
\usepackage{mathdots}
\usepackage{graphicx}
\usepackage[unicode=true,pdfusetitle,
 bookmarks=true,bookmarksnumbered=true,bookmarksopen=true,bookmarksopenlevel=1,
 breaklinks=false,pdfborder={0 0 1},backref=section,colorlinks=false]
 {hyperref}
\hypersetup{
 hypertex,bookmarksnumbered=true}

\makeatletter
\numberwithin{equation}{section}
\numberwithin{figure}{section}
\theoremstyle{plain}
\newtheorem{thm}{\protect\theoremname}
\theoremstyle{remark}
\newtheorem{rem}[thm]{\protect\remarkname}
\theoremstyle{plain}
\newtheorem{prop}[thm]{\protect\propositionname}

\makeatother

\providecommand{\propositionname}{Proposition}
\providecommand{\remarkname}{Remark}
\providecommand{\theoremname}{Theorem}

\begin{document}
\title{Perturbations of circuit evolution matrices with Jordan blocks}
\author{Alexander Figotin}
\address{University of California at Irvine, CA 92967}
\begin{abstract}
In our prior studies we synthesized special circuits possessing evolution
matrices that involve nontrivial Jordan blocks and the corresponding
degenerate eigenfrequencies. The degeneracy of this type is sometimes
referred to as exceptional point of degeneracy (EPD). The simplest
of these circuits are composed just of two $LC$-loops coupled by
a gyrator and they are of our primary interest here. These simple
circuits when near an EPD state can be used for enhanced sensitivity
applications. With that in mind we develop here a comprehensive perturbation
theory for these simple circuits near an EPD as well way to assure
their stable operation.

As to broader problem of numerical treatment of Jordan blocks and
their perturbation we propose a few approaches allowing to detect
the proximity to Jordan blocks.
\end{abstract}

\keywords{Electric circuit, exceptional point of degeneracy (EPD), Jordan block,
perturbations, instability, sensitivity.}
\maketitle

\section{Introduction\label{sec:intro}}

There is growing interest to electromagnetic system exhibiting Jordan
eigenvector degeneracy, which is a degeneracy of the system evolution
matrix when not only some eigenvalues, which often are frequencies,
coincide but the corresponding eigenvectors coincide also. The degeneracy
of this type is sometimes referred to as exceptional point of degeneracy
(EPD), \cite[II.1]{Kato}. As to the Jordan canonical forms theory
and its relations to differential equations we refer the reader to
\cite[III.4]{Hale}, \cite[3.1,3.2]{HorJohn}, \cite[25.4]{ArnODE}.

A particularly important class of applications of EPDs in sensing,
\cite{CheN}. \cite{HHWGECK}, \cite{KNAC}, \cite{PeLiXu}, \cite{Wie},
\cite{Wie1}. The idea, in nutshell, of using EPDs for sensing is
as follows. Suppose that sensing system is set to be at an EPD point
before the measurement. Then when a sensor is engaged its signal of
presumably small amplitude $\epsilon$ perturbs the system altering
its eigenfrequencies proportionally to $\sqrt{\epsilon}$. Since for
small $\epsilon$ its square root $\sqrt{\epsilon}$ is much larger
than $\epsilon$ we get enhanced sensitivity.

We study here those aspects of the perturbation theory of Jordan blocks
useful in detecting their presence and quantifying their proximity.
Motivated by enhanced sensing applications we also develop detailed
perturbation theory for the case of simple circuits as in Fig. \ref{fig:pri-cir2}
we introduced in \cite{FigSynbJ}. The typical circuit of interest
is composed of two $LC$-loops coupled with a gyrator, where $L_{j}$
and $C_{j}$ for $j=1,2$ and $G$ are respectively inductances and
capacitances, that cane be positive or negative, and the gyrator resistance.
We refer the list of these parameters as the circuit parameters. We
concisely review the relevant these circuits subjects from \cite{FigSynbJ}
in Section \ref{sec:circuit2} and develop the perturbation for the
circuits in Section \ref{sec:percir2}.

An infinitesimally small transformation can be make a matrix with
nontrivial Jordan structure completely diagonalizable. This fact complicates
numerical assessment of the presence of non-trivial Jordan blocks
as well as their utilization in application. In Section \ref{sec:Jordper}
we consider several approaches allowing to detect Jordan blocks numerically
and quantify proximity to them.

\begin{figure}[h]
\centering{}\includegraphics[scale=0.8]{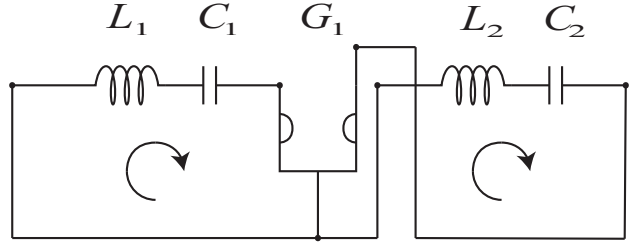}\caption{\label{fig:pri-cir2} For particular choices of values for quantities
$L_{1}$, $C_{1}$, $L_{2}$, $C_{2}$ and $G_{1}$ the evolution
matrix of this circuit develops EPDs, and its Jordan canonical form
consists of exactly two Jordan blocks of the size two. }
\end{figure}

\section{circuit and its specifications\label{sec:circuit2}}

We review concisely here properties of the circuit in Fig. \ref{fig:pri-cir2}
that we studied in \cite[4]{FigSynbJ}. The circuit vector evolution
equation and the corresponding eigenvalue problem are as follows
\begin{equation}
\partial_{t}\mathsf{q}=\mathscr{C}\mathsf{q},\quad\mathscr{C}=\left[\begin{array}{rrrr}
0 & 0 & 1 & 0\\
0 & 0 & 0 & 1\\
-\frac{1}{L_{1}C_{1}} & 0 & 0 & \frac{G_{1}}{L_{1}}\\
0 & -\frac{1}{L_{2}C_{2}} & -\frac{G_{1}}{L_{2}} & 0
\end{array}\right],\quad\mathsf{q}=\left[\begin{array}{c}
q\\
\partial_{t}q
\end{array}\right],\label{eq:Adqs1dp}
\end{equation}
\begin{equation}
\left(s\mathbb{I}-\mathscr{C}\right)\mathsf{q}=0,\quad\mathsf{q}=\left[\begin{array}{c}
q\\
sq
\end{array}\right].\label{eq:Adqs1ep}
\end{equation}
The characteristic polynomial associated with matrix $\mathscr{C}$
defined by equations (\ref{eq:Adqs1dp}) and the corresponding characteristic
equations are respectively
\begin{equation}
\chi\left(s\right)=\det\left\{ s\mathbb{I}-\mathscr{C}\right\} =s^{4}+\left(\xi_{1}+\xi_{2}+\frac{g}{L_{1}L_{2}}\right)s^{2}+\xi_{1}\xi_{2},\label{eq:Adqs2ap}
\end{equation}
\begin{equation}
\chi\left(s\right)=s^{4}+\left(\xi_{1}+\xi_{2}+\frac{g}{L_{1}L_{2}}\right)s^{2}+\xi_{1}\xi_{2}=0,\label{eq:Adqs2asp}
\end{equation}
where
\begin{equation}
\xi_{1}=\frac{1}{L_{1}C_{1}},\quad\xi_{2}=\frac{1}{L_{2}C_{2}},\quad g=G_{1}^{2}.\label{eq:Adqs2bp}
\end{equation}
We refer to positive $g=G_{1}^{2}$ in equations (\ref{eq:Adqs2bp})
as the \emph{gyration parameter}. Notice that using parameters in
equations (\ref{eq:Adqs2bp}) we can recast the companion matrix $\mathscr{C}$
defined by equations (\ref{eq:Adqs1dp}) and its characteristic function
$\chi\left(s\right)$ as in equation (\ref{eq:Adqs2asp}) as follows
\begin{equation}
\mathscr{C}=\left[\begin{array}{rrrr}
0 & 0 & 1 & 0\\
0 & 0 & 0 & 1\\
-\xi_{1} & 0 & 0 & \frac{G_{1}}{L_{1}}\\
0 & -\xi_{2} & -\frac{G_{1}}{L_{2}} & 0
\end{array}\right],\label{eq:Adqs2cp}
\end{equation}
\begin{equation}
\chi\left(s\right)=\chi_{h}=h^{2}+\left(\xi_{1}+\xi_{2}+\frac{g}{L_{1}L_{2}}\right)h+\xi_{1}\xi_{2},\quad h=s^{2}.\label{eq:Adqs2dp}
\end{equation}
The solutions to the quadratic equation $\chi_{h}=0$ are
\begin{equation}
h_{\pm}=\frac{-\left(\xi_{1}+\xi_{2}+\frac{g}{L_{1}L_{2}}\right)\pm\sqrt{\Delta_{h}}}{2},\label{eq:Adqs2ehp}
\end{equation}
where
\begin{equation}
\Delta_{h}=\frac{g^{2}}{L_{1}^{2}L_{2}^{2}}+\frac{2\left(\xi_{1}+\xi_{2}\right)g}{L_{1}L_{2}}+\left(\xi_{1}-\xi_{2}\right)^{2}\label{eq:Adqs2ep}
\end{equation}
is the discriminant of the quadratic polynomial $\chi_{h}$ (\ref{eq:Adqs2dp}).
The corresponding four solutions $s$ to the characteristic equation
(\ref{eq:Adqs2asp}), that is the eigenvalues, are
\begin{equation}
s=\pm\sqrt{h_{+}},\pm\sqrt{h_{-}},\label{eq:Adqs2ehsp}
\end{equation}
where $h_{\pm}$ satisfy equations (\ref{eq:Adqs2ehp}).

Notice that the eigenvalue degeneracy condition turns into equation
$\Delta_{h}=0$ which is equivalent to
\begin{equation}
L_{1}^{2}L_{2}^{2}\Delta_{h}=g^{2}+2\left(\xi_{1}+\xi_{2}\right)gL_{1}L_{2}+\left(\xi_{1}-\xi_{2}\right)^{2}L_{1}^{2}L_{2}^{2}=0.\label{eq:Adqs3ap}
\end{equation}
Equation (\ref{eq:Adqs3ap}) evidently is a constraint on the circuit
parameters which is a quadratic equation for $g$. Being given the
circuit coefficients $\xi_{1}$, $\xi_{2}$, $L_{1}$ and $L_{2}$
this quadratic in $g$ equation has exactly two solutions
\begin{equation}
g_{\delta}=\left(-\xi_{1}-\xi_{2}+2\delta\sqrt{\xi_{1}\xi_{2}}\right)L_{1}L_{2},\quad\delta=\pm1.\label{eq:Adqs3bp}
\end{equation}
We refer to $g_{\delta}$ in equations (\ref{eq:Adqs3bp}) as\emph{
special values of the gyration parameter $g$}. For the two special
values $g$ we get from equations (\ref{eq:Adqs2ehp}) the corresponding
two degenerate roots
\begin{equation}
h=-\frac{\xi_{1}+\xi_{2}+\frac{\dot{g}}{L_{1}L_{2}}}{2}=\pm\sqrt{\xi_{1}\xi_{2}}.\label{eq:Adqs3bhp}
\end{equation}
Since $G_{1}$ is real then $g=G_{1}^{2}$ is real as well. The expression
(\ref{eq:Adqs3bp}) for $g$ is real-valued if and only if
\begin{equation}
\xi_{1}\xi_{2}>0,\text{ or equivalently }\mathrm{\frac{\xi_{1}}{\left|\xi_{1}\right|}=\frac{\xi_{2}}{\left|\xi_{2}\right|}}=\sigma,\label{eq:Adqs3cp}
\end{equation}
where we introduced a binary variable $\sigma$ taking values $\pm1$.
We refer to $\sigma$ as the \emph{circuit sign index}. Relations
(\ref{eq:Adqs3cp}) imply in particular that the equality of signs
$\mathrm{sign}\,\left\{ \xi_{1}\right\} =\mathrm{sign}\,\left\{ \xi_{2}\right\} $
is a necessary condition for the eigenvalue degeneracy condition $\Delta_{h}=0$
provided that $g$ has to be real-valued.

It follows then from relations (\ref{eq:Adqs3bp}) and (\ref{eq:Adqs3cp})
that the special values of the gyration parameter $g_{\delta}$ can
be recast as
\begin{gather}
g_{\delta}=-\sigma\left(\sqrt{\left|\xi_{1}\right|}+\delta\sqrt{\left|\xi_{2}\right|}\right)^{2}L_{1}L_{2},\quad\delta=\pm1,\label{eq:Adqs3dp}
\end{gather}
where $\sqrt{\xi}>0$ for $\xi>0$. Recall that $g=G_{1}^{2}>0$ and
to provide for that we must have in right-hand side of equations (\ref{eq:Adqs3dp})
\begin{equation}
-\sigma L_{1}L_{2}>0\text{, or equivalently }\frac{L_{1}L_{2}}{\left|L_{1}L_{2}\right|}=-\sigma.\label{eq:Adqs3ep}
\end{equation}
Relations (\ref{eq:Adqs3cp}) and (\ref{eq:Adqs3ep}) on the signs
of the involved parameters can be combined into the \emph{circuit
sign constraints}
\begin{equation}
\mathrm{sign}\,\left\{ \xi_{1}\right\} =\mathrm{sign}\,\left\{ \xi_{2}\right\} =-\mathrm{sign}\,\left\{ L_{1}L_{2}\right\} =\mathrm{sign}\,\left\{ \sigma\right\} .\label{eq:Adqs3esp}
\end{equation}
\emph{Notice that the sign constraints (\ref{eq:Adqs3esp}) involving
the circuit index $\sigma$ defined by (\ref{eq:Adqs3cp}) are necessary
for the eigenvalue degeneracy condition $\Delta_{h}=0$. }Combing
equations (\ref{eq:Adqs3dp}) and (\ref{eq:Adqs3ep}) we obtain
\begin{gather}
g_{\delta}=\left(\sqrt{\left|\xi_{1}\right|}+\delta\sqrt{\left|\xi_{2}\right|}\right)^{2}\left|L_{1}L_{2}\right|,\quad\delta=\pm1,\text{ assuming the circuit sign constraints}.\label{eq:Adqs3gp}
\end{gather}
Since $g=G_{1}^{2}>0$ the special values of the gyrator resistance
$\dot{G}_{1}$ corresponding to the special values $g_{\delta}$ as
in equation (\ref{eq:Adqs3gp}) are
\begin{gather}
\dot{G}_{1}=\sigma_{1}\sqrt{g_{\delta}}=\sigma_{1}\left(\sqrt{\left|\xi_{1}\right|}+\delta\sqrt{\left|\xi_{2}\right|}\right)\sqrt{\left|L_{1}L_{2}\right|},\text{ assuming the circuit sign constraints, }\label{eq:Adqs3hp}
\end{gather}
where the binary variable $\sigma_{1}$ takes values $\pm1$.

Using representation (\ref{eq:Adqs3gp}) for $g_{\delta}$ under the
circuit sign constraints (\ref{eq:Adqs3esp}) we can recast the expression
for the degenerate root $\dot{h}$ in equations (\ref{eq:Adqs3bhp})
as follows
\begin{equation}
h_{\delta}=\sigma\delta\sqrt{\left|\xi_{1}\right|}\sqrt{\left|\xi_{2}\right|},\text{for }g=g_{\delta}=\left(\sqrt{\left|\xi_{1}\right|}+\delta\sqrt{\left|\xi_{2}\right|}\right)^{2}\left|L_{1}L_{2}\right|,\quad\delta=\pm1,\label{eq:Adqs3fp}
\end{equation}
where $\sigma$ is the circuit sign index defined by equations (\ref{eq:Adqs3cp})
and $\sqrt{\xi}>0$ for $\xi>0$.

One the principal results regarding the circuit in Fig. \ref{fig:pri-cir2}
is as follows \cite{FigSynbJ} .
\begin{thm}[Jordan form under degeneracy]
\label{thm:jordeg} Let the circuit be as depicted in Fig. \ref{fig:pri-cir2}
and let all its parameters $L_{1}$, $C_{1}$, $L_{2}$, $C_{2}$
and $G_{1}$ be real and non-zero. Then the companion matrix $\mathscr{C}$
satisfying equations (\ref{eq:Adqs1dp}) and (\ref{eq:Adqs2cp}) has
the Jordan form
\begin{equation}
\mathscr{J}=\left[\begin{array}{rrrr}
s_{0} & 1 & 0 & 0\\
0 & s_{0} & 0 & 0\\
0 & 0 & -s_{0} & 1\\
0 & 0 & 0 & -s_{0}
\end{array}\right],\label{eq:Adqs4cp}
\end{equation}
if and only if the circuit parameters satisfy the degeneracy conditions
as in equation (\ref{eq:Adqs3ap}). Then for $g=g_{\delta}$ we have
\begin{equation}
h_{\delta}=\sigma\delta\sqrt{\left|\xi_{1}\right|}\sqrt{\left|\xi_{2}\right|},\text{ for }g=g_{\delta}=\left(\sqrt{\left|\xi_{1}\right|}+\delta\sqrt{\left|\xi_{2}\right|}\right)^{2}\left|L_{1}L_{2}\right|,\label{eq:Adqs4cap}
\end{equation}
\begin{equation}
\pm s_{0}=\pm\sqrt{\sigma\delta\sqrt{\left|\xi_{1}\right|}\sqrt{\left|\xi_{2}\right|}},\text{ for }g=g_{\delta}=\left(\sqrt{\left|\xi_{1}\right|}+\delta\sqrt{\left|\xi_{2}\right|}\right)^{2}\left|L_{1}L_{2}\right|,\label{eq:Adqs4cbp}
\end{equation}
where $\sqrt{\xi}>0$ for $\xi>0$, $\delta=\pm1$ and $\sigma$ is
the circuit sign index defined by equations (\ref{eq:Adqs3cp}). According
to formula (\ref{eq:Adqs4cbp}) degenerate eigenvalues $\pm s_{0}$
depend on the product $\sigma\delta$ and $\left|\xi_{1}\right|,\left|\xi_{2}\right|$
and consequently they are either real or pure imaginary depending
on whether $\delta=\sigma$ or $\delta=-\sigma$.

Notice that in the special case when $\left|\xi_{1}\right|=\left|\xi_{2}\right|$
the parameter $g_{\delta}$ takes only one non-zero value, namely
\begin{equation}
g=g_{1}=4\left|\xi_{1}\right|\left|L_{1}L_{2}\right|,\quad L_{1}L_{2}=-\sigma\left|L_{1}L_{2}\right|,\label{eq:Adqs4dp}
\end{equation}
whereas $g_{-1}=0$ which is inconsistent with our assumption $G_{1}\neq0$.
Evidently for $g=0$ the circuit breaks into two independent $LC$-circuits
and in this case the relevant Jordan form is a diagonal $4\times4$
matrix with eigenvalues $\pm\sqrt{\xi_{1}}$ and $\pm\sqrt{\xi_{2}}$
.
\end{thm}

\begin{rem}[instability and marginal stability]
\label{rem:stability} Notice according to formula (\ref{eq:Adqs4cbp})
degenerate eigenvalues $\pm s_{0}$ are real for $\delta=\sigma$
and hence they correspond to exponentially growing and decaying in
time solutions indicating instability. For $\delta=-\sigma$ the degenerate
eigenvalues $\pm s_{0}$ are pure imaginary corresponding to oscillatory
solutions indicating that there is at least marginal stability.
\end{rem}

To get a graphical illustration for the circuit complex-valued eigenvalues
as functions of the gyration parameter $g$ we use the following data
\begin{equation}
\left|\xi_{1}\right|=1,\quad\left|\xi_{2}\right|=2,\quad\left|L_{1}\right|=1,\quad\left|L_{2}\right|=2,\quad\sigma=1,\label{eq:xiLdat1ap}
\end{equation}
 and that corresponds to
\begin{equation}
\xi_{1}=1,\quad\xi_{2}=2,\quad L_{1}=\pm1,\quad L_{2}=\mp2.\label{eq:xiLdat1bp}
\end{equation}
It follows then from representation (\ref{eq:Adqs3gp}) that the corresponding
special values $g_{\delta}$ are
\begin{equation}
g_{-1}=2\left(1-\sqrt{2}\right)^{2}\cong0.3431457498,\quad g_{1}=2\left(1+\sqrt{2}\right)^{2}\cong11.65685425.\label{eq:xiLdat1cp}
\end{equation}
\begin{figure}[h]
\centering{}\includegraphics[scale=0.5]{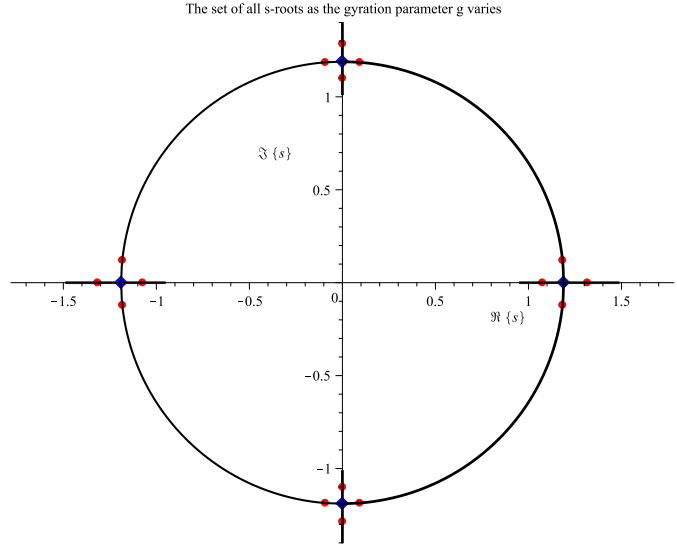}\caption{\label{fig:eigenvals} The plot shows the set $S_{\mathrm{eig}}$
of all complex valued eigenvalues $s$ defined by equations (\ref{eq:Adqs2ehsp})
for the data in equations (\ref{eq:xiLdat1ap}), (\ref{eq:xiLdat1bp})
when the gyration parameter $g$ varies in interval containing special
values $g_{-1}$ and $g_{1}$ defined in relations (\ref{eq:xiLdat1cp}).
The horizontal and vertical axes represent the real and the imaginary
parts $\Re\left\{ s\right\} $ and $\Im\left\{ s\right\} $ of eigenvalues
$s$. $S_{\mathrm{eig}}$ consists of the circle centered in the origin
of radius $\sqrt[4]{\left|\xi_{1}\right|\left|\xi_{2}\right|}$ and
four intersecting its intervals lying on real and imaginary axes.
The circular part of the set $S_{\mathrm{eig}}$ corresponds to all
eigenvalues for $g_{-1}\protect\leq g\protect\leq$$g_{1}$. The degenerate
eigenvalues $\pm s_{0}$ corresponding to $g_{-1}$ and $g_{1}$ and
defined by equations (\ref{eq:Adqs4cbp}) are shown as solid diamond
(blue) dots. Two of them are real, positive and negative, numbers
and another two are pure imaginary, with positive and negative imaginary
parts. The 16 solid circular (red) dots are associated with 4 quadruples
of eigenvalues corresponding to 4 different values of the gyration
parameter $g$ chosen to be slightly larger or smaller than the special
values $g_{-1}$ and $g_{1}$. Let us take a look at any of the degenerate
eigenvalues identified by solid diamond (blue) dots. If $g$ is slightly
different from its special values $g_{-1}$ and $g_{1}$ then each
degenerate eigenvalue point splits into a pair of points identified
by solid circular (red) dots. They are either two real or two pure
imaginary points if $g$ is outside the interval $\left[g_{-1},g_{1}\right]$,
or alternatively they are two points lying on the circle, if $g$
is inside the interval $\left[g_{-1},g_{1}\right]$.}
\end{figure}

To explain the rise of the circular part of the set $S_{\mathrm{eig}}$
in Fig. \ref{fig:eigenvals} we recast the characteristic equation
(\ref{eq:Adqs2dp}) as follows
\begin{equation}
H+\frac{1}{H}=R,\quad R=\frac{\sigma}{\sqrt{\left|\xi_{1}\right|\left|\xi_{2}\right|}}\left[\frac{g}{\left|L_{1}\right|\left|L_{2}\right|}-\left(\left|\xi_{1}\right|+\left|\xi_{2}\right|\right)\right],\quad H=\frac{h}{\sqrt{\left|\xi_{1}\right|\left|\xi_{2}\right|}}.\label{eq:charHR1ap}
\end{equation}
Notice that
\begin{equation}
R=2\delta\sigma,\text{for }g=g_{\delta}=\left(\sqrt{\left|\xi_{1}\right|}+\delta\sqrt{\left|\xi_{2}\right|}\right)^{2}\left|L_{1}L_{2}\right|,\quad\delta=\pm1,\label{eq:charHR1bp}
\end{equation}
where $\sigma=\pm1$ is the circuit sign index. Since $R$ depends
linearly on $g$ relations (\ref{eq:charHR1ap}) and (\ref{eq:charHR1bp})
imply
\begin{equation}
\left|R\right|\leq2,\text{ for }g_{-1}\leq g\leq g_{1};\quad\left|R\right|>2,\text{ for }g<g_{-1}\text{ and }g>g_{1}.\label{eq:charHR1c}
\end{equation}
It is an elementary fact that solutions $H$ to equation (\ref{eq:charHR1ap})
satisfy the following relations:
\begin{equation}
H=\exp\left\{ \mathrm{i}\theta\right\} ,\text{ for }\left|R\right|\leq2,\text{ and }\cos\left(\theta\right)=\frac{R}{2},\quad0\leq\theta\leq\pi,\label{eq:charHR1dp}
\end{equation}
\begin{equation}
H>0,\text{ for }R>2,\text{ and }H<0,\text{ for }R<-2.\label{eq:charHR1ep}
\end{equation}
It is also evident from the form of equation (\ref{eq:charHR1ap})
that if $H$ is its solution then $H^{-1}$ is a solution as well,
that is the two solutions to equation (\ref{eq:charHR1ap}) always
come in pairs of the form $\left\{ H,H^{-1}\right\} $.

Since the eigenvalues $s$ satisfy $s=\pm\sqrt{h}$ the established
above properties of $h=\sqrt{\left|\xi_{1}\right|\left|\xi_{2}\right|}H$
can recast for $s$ as follows.
\begin{thm}[quadruples of eigenvalues]
\label{thm:quad} For every $g>0$ every solution $s$ to the characteristic
equation (\ref{eq:Adqs2asp}) is of the form (\ref{eq:Adqs2ehsp})
and the number of solutions is exactly four counting their multiplicity.
Every such a quadruple of solutions is of the following form
\begin{equation}
\left\{ s,\frac{\sqrt{\left|\xi_{1}\right|\left|\xi_{2}\right|}}{s},-s,-\frac{\sqrt{\left|\xi_{1}\right|\left|\xi_{2}\right|}}{s}\right\} ,\label{eq:charHR2ap}
\end{equation}
where $s$ is a solution to the characteristic equation (\ref{eq:Adqs2asp}.
Then for $g_{-1}\leq g\leq g_{1}$ the quadruple of solutions belongs
to the circle $\left|s\right|=\sqrt[4]{\left|\xi_{1}\right|\left|\xi_{2}\right|}$
such that
\begin{equation}
s=\delta_{1}\sqrt[4]{\left|\xi_{1}\right|\left|\xi_{2}\right|}\exp\left\{ \mathrm{i}\delta_{2}\theta\right\} ,\cos\left(\theta\right)=\frac{R}{2},\quad0\leq\theta\leq\pi,\delta_{1},\delta_{2}=\pm1,\label{eq:charHR2bp}
\end{equation}
where $R$ is defined in relations (\ref{eq:charHR1ap}). If $g<g_{-1}$
or $g>g_{1}$ the quadruple of solutions consists of either real numbers
and pure imaginary numbers depending if $R>2$ or $R<-2$ respectively.
In view of relations (\ref{eq:charHR1ap}) and $s=\pm\sqrt{h}$ where
$h=\sqrt{\left|\xi_{1}\right|\left|\xi_{2}\right|}H$ every quadruple
of solutions as in expression (\ref{eq:charHR2ap}) is invariant with
respect to the complex conjugation transformation.
\end{thm}

The following remark discusses in some detail the transition of eigenvalues
lying on the circle $\left|s\right|=\sqrt[4]{\left|\xi_{1}\right|\left|\xi_{2}\right|}$
having non-zero real and imaginary parts into either real or pure
imaginary numbers as the value of the gyration parameter $g$ passes
through its special values $g_{-1}$ or $g_{1}$ at which the eigenvalues
degenerate.
\begin{rem}[transition at degeneracy points]
 \label{rem:degpoint}According to formula (\ref{eq:Adqs4cbp}) there
is total of four degenerate eigenvalues $\pm s_{0}$, namely $\pm$$\sqrt[4]{\left|\xi_{1}\right|\left|\xi_{2}\right|}$
and $\pm\mathrm{i}$$\sqrt[4]{\left|\xi_{1}\right|\left|\xi_{2}\right|}$
(depicted as solid diamond (blue) dots in Fig. \ref{fig:eigenvals})
that are associated with the two special values of the gyration parameter
$g_{\pm1}=\left(\sqrt{\left|\xi_{1}\right|}\pm\sqrt{\left|\xi_{2}\right|}\right)^{2}\left|L_{1}L_{2}\right|$.
For any value of the gyration parameter $g$ different than its two
special values there are exactly four distinct eigenvalues $s$ forming
a quadruple as in expression (\ref{eq:charHR2ap}). If $g_{-1}<g<g_{1}$
and $g$ gets close to either $g_{-1}$ or $g_{1}$ the corresponding
four distinct eigenvalues on the circle $\left|s\right|=\sqrt[4]{\left|\xi_{1}\right|\left|\xi_{2}\right|}$
get close to either $\pm$$\sqrt[4]{\left|\xi_{1}\right|\left|\xi_{2}\right|}$
or $\pm\mathrm{i}$$\sqrt[4]{\left|\xi_{1}\right|\left|\xi_{2}\right|}$
as depicted in Fig. \ref{fig:eigenvals}) by solid circle (red) dots.
As $g$ approaches the special values $g_{-1}$ or $g_{1}$, reaches
them and gets out of the interval $\left[g_{-1},g_{1}\right]$ the
corresponding solid circle (red) dots approach the relevant points
$\pm$$\sqrt[4]{\left|\xi_{1}\right|\left|\xi_{2}\right|}$ or $\pm\mathrm{i}$$\sqrt[4]{\left|\xi_{1}\right|\left|\xi_{2}\right|}$
, merge at them and then split again passing to respectively real
and imaginary axes as illustrated by Fig. \ref{fig:eigenvals}.
\end{rem}

\section{Perturbation theory for the circuit\label{sec:percir2}}

We develop here the perturbation theory for the circuit in Fig. \ref{fig:pri-cir2}
at its EPD point. The circuit and its relevant properties are reviewed
in Section \ref{sec:circuit2}. Our primary case of interest here
is when the circuit parameters satisfy
\begin{equation}
\xi_{1},\xi_{2},-L_{1}L_{2}>0,\quad R\cong-2,\label{eq:xixiL1a}
\end{equation}
implying in view of the sign constraints (\ref{eq:Adqs2bp}) and relations
(\ref{eq:charHR1bp}) the following relations for the circuit parameters
at the EPD point:
\begin{equation}
\sigma=1,\quad\delta=-1;\quad R=2\delta\sigma=-2,\text{for }g=g_{_{-1}}=-\left(\sqrt{\xi_{1}}-\sqrt{\xi_{2}}\right)^{2}L_{1}L_{2}.\label{eq:xixiL1b}
\end{equation}
Then the frequencies
\begin{equation}
\omega_{1}=\frac{1}{\sqrt{L_{1}C_{1}}}=\sqrt{\xi_{1}}>0,\quad\omega_{2}=\frac{1}{\sqrt{L_{2}C_{2}}}=\sqrt{\xi_{2}}>0.\label{eq:xixiL1c}
\end{equation}
 associated with the circuit $LC$-loops are real, and also according
to Theorem \ref{thm:quad} and Remark \ref{rem:degpoint} the circuit
degenerate frequency at the EPD is
\begin{equation}
\dot{\omega}=\sqrt[4]{\xi_{1}\xi_{2}}=\sqrt{\omega_{1}\omega_{2}}>0,\label{eq:xixiL1ca}
\end{equation}
where $\xi_{1},\xi_{2}$ and $\omega_{1},\omega_{2}$ are the values
of the parameters ad the EPD point. Equation (\ref{eq:xixiL1ca})
indicate that the circuit at the EPD is at least marginally stable
explaining why this EPD point satisfying equations (\ref{eq:xixiL1a})
and (\ref{eq:xixiL1b}) is of our primary interest.

Let us turn now to the characteristic equation (\ref{eq:charHR1ap}).
This equation under conditions (\ref{eq:xixiL1a}) and in view of
relations (\ref{eq:xixiL1b}) takes here the form
\begin{gather}
H+\frac{1}{H}=-2-r,\nonumber \\
r=\frac{1}{\sqrt{\xi_{1}\xi_{2}}}\left[\frac{g}{L_{1}L_{2}}+\left(\xi_{1}+\xi_{2}\right)\right]+2,\quad H=\frac{h}{\sqrt{\xi_{1}\xi_{2}}},\label{eq:xixiL1d}
\end{gather}
where
\begin{equation}
r=0,\text{for }g=g_{_{-1}}=-\left(\sqrt{\xi_{1}}-\sqrt{\xi_{2}}\right)^{2}L_{1}L_{2}.\label{eq:xixiL1e}
\end{equation}
Let us consider now a fixed EPD point with the circuit parameters
\begin{equation}
\mathring{L}_{j},\quad\mathring{C}_{j},\quad j=1,2;\quad\mathring{g}=-\left(\sqrt{\dot{\xi}_{1}}-\sqrt{\dot{\xi}_{2}}\right)^{2}\dot{L}_{1}\dot{L}_{2}\label{eq:xixiL2a}
\end{equation}
and its perturbation of the form

\begin{equation}
L_{j}=\mathring{L}_{j}\left(1+\Delta\left(L_{j}\right)\epsilon\right),\quad C_{j}=\mathring{C}_{j}\left(1+\Delta\left(C_{j}\right)\epsilon\right),\quad j=1,2;\quad g=\mathring{g}\left(1+\Delta\left(g\right)\epsilon\right),\label{eq:xixiL2b}
\end{equation}
where $\epsilon$ is a small real-valued parameter and $\Delta\left(L_{j}\right)$,
$\Delta\left(C_{j}\right)$ and $\Delta\left(g\right)$ are real-valued
coefficients that weight the variation of the . We often will assume
that it is only one of these coefficients differs from zero.

Note that if $r$ is small then the two solutions $H$ to equation
(\ref{eq:xixiL1d}) satisfy the following asymptotic formula
\begin{equation}
H=-1\pm\sqrt{r}+O\left(r\right),\quad r\rightarrow0.\label{eq:xixiL2c}
\end{equation}
For the circuit parameters as in equations (\ref{eq:xixiL2b}) we
have
\begin{equation}
r=\Delta\left(R\right)\epsilon+O\left(\epsilon^{2}\right),\quad\epsilon\rightarrow0,\label{eq:xixiL2d}
\end{equation}
where
\begin{gather}
\Delta\left(R\right)=\left(\frac{\sqrt{\xi_{1}}}{\sqrt{\xi_{2}}}-1\right)\Delta_{1}\left(R\right)+\left(\frac{\sqrt{\xi_{2}}}{\sqrt{\xi_{1}}}-1\right)\Delta_{2}\left(R\right),\label{eq:xixiL2e}\\
\Delta_{1}\left(R\right)=\Delta\left(L_{2}\right)-\Delta\left(C_{1}\right)+\Delta\left(g\right),\quad\Delta_{2}\left(R\right)=\Delta\left(L_{1}\right)-\Delta\left(C_{2}\right)+\Delta\left(g\right).\nonumber 
\end{gather}
Using equation $H=\frac{h}{\sqrt{\xi_{1}\xi_{2}}}$ and taking the
square root of the two solutions $H$ defined by equations (\ref{eq:xixiL2c})
and (\ref{eq:xixiL2d}) we obtain after elementary algebraic transformation
four solutions $s$ for the original characteristic equation (\ref{eq:Adqs2asp}),
(\ref{eq:Adqs2bp}) 
\begin{equation}
s=\mathrm{i}\omega=\mathrm{i}\sqrt[4]{\dot{\xi}_{1}\dot{\xi}_{2}}\left(1\pm\frac{1}{2}\sqrt{\Delta\left(R\right)\epsilon}+O\left(\epsilon\right)\right),\;-\mathrm{i}\sqrt{\dot{\xi}_{1}\dot{\xi}_{2}}\left(1\pm\frac{1}{2}\sqrt{\Delta\left(R\right)\epsilon}+O\left(\epsilon\right)\right),\quad\epsilon\rightarrow0,\label{eq:xixiL3a}
\end{equation}
where coefficient $\Delta\left(R\right)$ satisfies representation
(\ref{eq:xixiL2e}). Evidently the above solutions are pure imaginary
if $\Delta\left(R\right)\epsilon>0$.

Equations (\ref{eq:xixiL3a}) and (\ref{eq:xixiL1ca}) imply the following
expressions for the two split frequencies and their relative difference
\begin{equation}
\omega_{\pm}=\sqrt[4]{\dot{\xi}_{1}\dot{\xi}_{2}}\left(1\pm\frac{1}{2}\sqrt{\Delta\left(R\right)\epsilon}+O\left(\epsilon\right)\right)=\dot{\omega}\left(1\pm\frac{1}{2}\sqrt{\Delta\left(R\right)\epsilon}+O\left(\epsilon\right)\right),\label{eq:xixiL3aa}
\end{equation}
\begin{equation}
\frac{\omega_{+}-\omega_{-}}{\dot{\omega}}=\sqrt{\Delta\left(R\right)\epsilon}+O\left(\epsilon\right).\label{eq:xixiL3ab}
\end{equation}
We refer to the difference $\omega_{+}-\omega_{-}$ as \emph{frequency
split} and to $\frac{\omega_{+}-\omega_{-}}{\dot{\omega}}$ as \emph{relative
frequency split}.

Using frequencies $\omega_{1}$ and $\omega_{2}$ defined by equations
(\ref{eq:xixiL1c}) the representation (\ref{eq:xixiL2e}) for $\Delta\left(R\right)$
can be recast as

\begin{gather}
\Delta\left(R\right)=\left(\frac{\omega_{1}}{\omega_{2}}-1\right)\Delta_{1}\left(R\right)+\left(\frac{\omega_{2}}{\omega_{1}}-1\right)\Delta_{2}\left(R\right),\label{eq:xixiL3b}\\
\Delta_{1}\left(R\right)=\Delta\left(C_{1}\right)-\Delta\left(L_{2}\right)-\Delta\left(g\right),\quad\Delta_{2}\left(R\right)=\Delta\left(C_{2}\right)-\Delta\left(L_{1}\right)-\Delta\left(g\right).\nonumber 
\end{gather}
Notice that the following elementary inequality holds
\begin{equation}
\xi-1>1-\frac{1}{\xi}=\frac{\xi-1}{\xi},\text{ for }\xi>1,\label{eq:xixiL3c}
\end{equation}
readily implying

\begin{equation}
\left|\frac{\omega_{1}}{\omega_{2}}-1\right|>\left|\frac{\omega_{2}}{\omega_{1}}-1\right|\text{ if }\omega_{1}>\omega_{2}.\label{eq:xixi3d}
\end{equation}
Notice also the equation (\ref{eq:xixiL1ca}) for the circuit degenerate
frequency $\dot{\omega}$ at the EPD implies
\begin{equation}
\omega_{1}>\dot{\omega}>\omega_{2}\text{ if }\omega_{1}>\omega_{2}.\label{eq:xixiL3e}
\end{equation}

Assuming that the only circuit parameter that varies is capacitance
$C_{1}$ and that its relative change $\Delta\left(C_{1}\right)=\frac{C_{1}-\dot{C}_{1}}{\dot{C}_{1}}$
is small we obtain from equations (\ref{eq:xixiL3ab}) for $\epsilon=1$
\begin{equation}
\frac{\omega_{+}-\omega_{-}}{\dot{\omega}}\cong\sqrt{-\left(\frac{\omega_{1}}{\omega_{2}}-1\right)\Delta\left(C_{1}\right)}=\sqrt{-\left(\frac{\omega_{1}}{\omega_{2}}-1\right)\frac{C_{1}-\dot{C}_{1}}{\dot{C}_{1}}}.\label{eq:xixiL5a}
\end{equation}
Similarly under the assumption of smallness of the relative change
of the varying variables $L_{1}$, $C_{2}$, $L_{2}$, and $g$ we
get respectively the following relations
\begin{equation}
\frac{\omega_{+}-\omega_{-}}{\dot{\omega}}\cong\sqrt{\left(\frac{\omega_{2}}{\omega_{1}}-1\right)\Delta\left(L_{1}\right)}=\sqrt{\left(\frac{\omega_{2}}{\omega_{1}}-1\right)\frac{L_{1}-\dot{L}_{1}}{\dot{L}_{1}}},\label{eq:xixiL5b}
\end{equation}
\begin{equation}
\frac{\omega_{+}-\omega_{-}}{\dot{\omega}}\cong\sqrt{-\left(\frac{\omega_{2}}{\omega_{1}}-1\right)\Delta\left(C_{2}\right)}=\sqrt{-\left(\frac{\omega_{2}}{\omega_{1}}-1\right)\frac{C_{2}-\dot{C}_{2}}{\dot{C}_{2}}},\label{eq:xixiL5c}
\end{equation}
\begin{equation}
\frac{\omega_{+}-\omega_{-}}{\dot{\omega}}\cong\sqrt{\left(\frac{\omega_{1}}{\omega_{2}}-1\right)\Delta\left(L_{2}\right)}=\sqrt{\left(\frac{\omega_{1}}{\omega_{2}}-1\right)\frac{L_{2}-\dot{L}_{2}}{\dot{L}_{2}}},\label{eq:xixiL5d}
\end{equation}
\begin{equation}
\frac{\omega_{+}-\omega_{-}}{\dot{\omega}}\cong\sqrt{\left(2-\frac{\omega_{1}}{\omega_{2}}-\frac{\omega_{2}}{\omega_{1}}\right)\Delta\left(g\right)}=\sqrt{\left(2-\frac{\omega_{1}}{\omega_{2}}-\frac{\omega_{2}}{\omega_{1}}\right)\frac{g-\dot{g}}{\dot{g}}}.\label{eq:xixiL5e}
\end{equation}
Notice that for the quantities under the sign of square root in relations
(\ref{eq:xixiL5a})-(\ref{eq:xixiL5e}) to be positive the relevant
relative changes of the circuit parameters must have proper signs.

In view of equations (\ref{eq:xixiL3a}), (\ref{eq:xixiL3b}) and
inequality (\ref{eq:xixiL3c}) as well as relations (\ref{eq:xixiL5a})-(\ref{eq:xixiL5e})
the following statement holds.
\begin{thm}[rate of change of the eigenfrequencies]
\label{thm:omDelR} Suppose that $\omega_{1}>\omega_{2}$ where the
indicated frequencies are defined by equations (\ref{eq:xixiL1c}).
Then the rate of change of the eigenfrequencies defined by equation
(\ref{eq:xixiL3a}) for small $\epsilon$ is higher as $C_{1}$ and
$L_{2}$ vary compare with the same as $C_{2}$ and $L_{1}$ vary.
Fig. \ref{fig:pertC12} provides graphical representation of stated
difference in the rate of change of the eigenfrequencies for for the
data as in equations (\ref{eq:xixiL4a}) and (\ref{eq:xixi4Lb}).
\end{thm}

To get a feel of the considered above perturbation theory let us consider
an example consistent with assumptions (\ref{eq:xixiL1b}) with the
following values of the circuit parameters for an EPD point
\begin{gather}
\dot{C}_{1}=1,\quad\dot{L}_{1}=1,\quad\dot{C}_{2}=-2,\quad\dot{L}_{2}=-1,\label{eq:xixiL4a}\\
\sigma=1,\quad\delta=-1;\quad\dot{R}=2\delta\sigma=-2,\text{for }\dot{g}=g_{_{-1}}=2\left(1-\sqrt{\frac{1}{4}}\right)^{2}=\frac{1}{2}.\nonumber 
\end{gather}
Consequently
\begin{equation}
\dot{\xi}_{1}=1,\quad\dot{\xi}_{2}=\frac{1}{4},\quad\dot{\omega}_{1}=1,\quad\dot{\omega}_{2}=\frac{1}{4},\quad\dot{\omega}_{0}=\frac{1}{2}.\label{eq:xixi4Lb}
\end{equation}

The Figs. \ref{fig:pertC1-1}, \ref{fig:pertC1-2}, \ref{fig:pertC2-1},
\ref{fig:pertC2-2}, \ref{fig:pertC12}, \ref{fig:pertg-1} and \ref{fig:pertg-2}
show the real and imaginary parts of the four eigenvalues $s$ as
in equations (\ref{eq:Adqs2ap}) assuming that (i) the circuit parameters
for the EPD point satisfy equations (\ref{eq:xixiL4a}) and (\ref{eq:xixi4Lb});
(ii) the perturbed circuit parameters satisfy equations (\ref{eq:xixiL2b})
where $\epsilon=1$ and the eigenvalues approximations to be as in
equations (\ref{eq:xixiL3a}) where we assume $\epsilon=1$ and $\Delta\left(R\right)$
is small.
\begin{figure}[h]
\centering{}\includegraphics[scale=0.4]{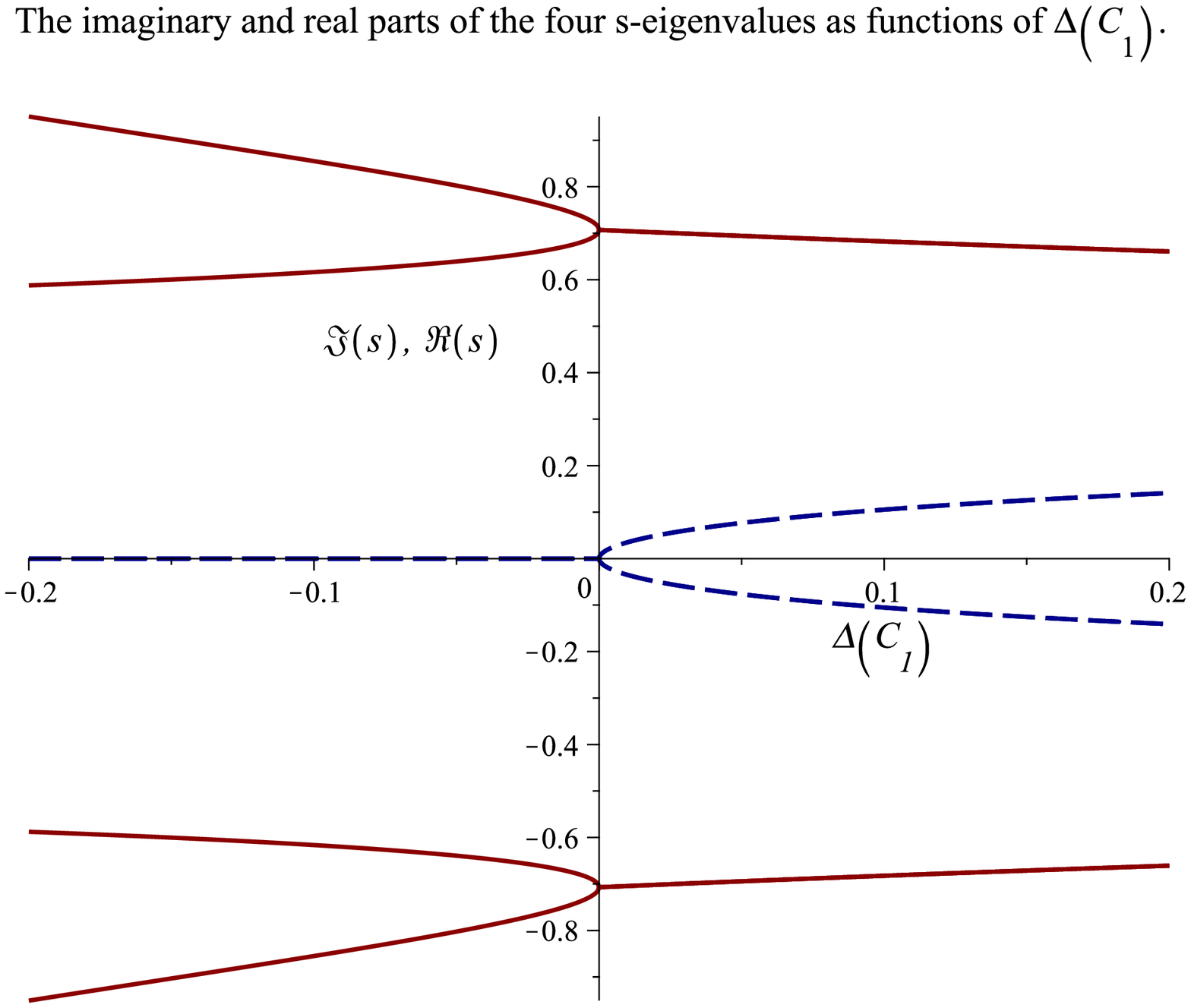}\caption{\label{fig:pertC1-1} This plot shows the real (blue dashed line)
and imaginary (solid brown line) parts of the four eigenvalues $s$
as in equations (\ref{eq:Adqs2ap}) assuming that (i) the circuit
parameters satisfy equations (\ref{eq:xixiL2b}) where $\epsilon=1$,
$\Delta\left(L_{1}\right)=\Delta\left(L_{2}\right)=\Delta\left(C_{2}\right)=\Delta\left(g\right)=0$
and $\Delta\left(C_{1}\right)$ varies. In other words, the plot shows
the variation of the real and imaginary parts of the four eigenvalues
$s$ as function on the capacitance $C_{1}$.}
\end{figure}
\begin{figure}[h]
\begin{centering}
\includegraphics[scale=0.4]{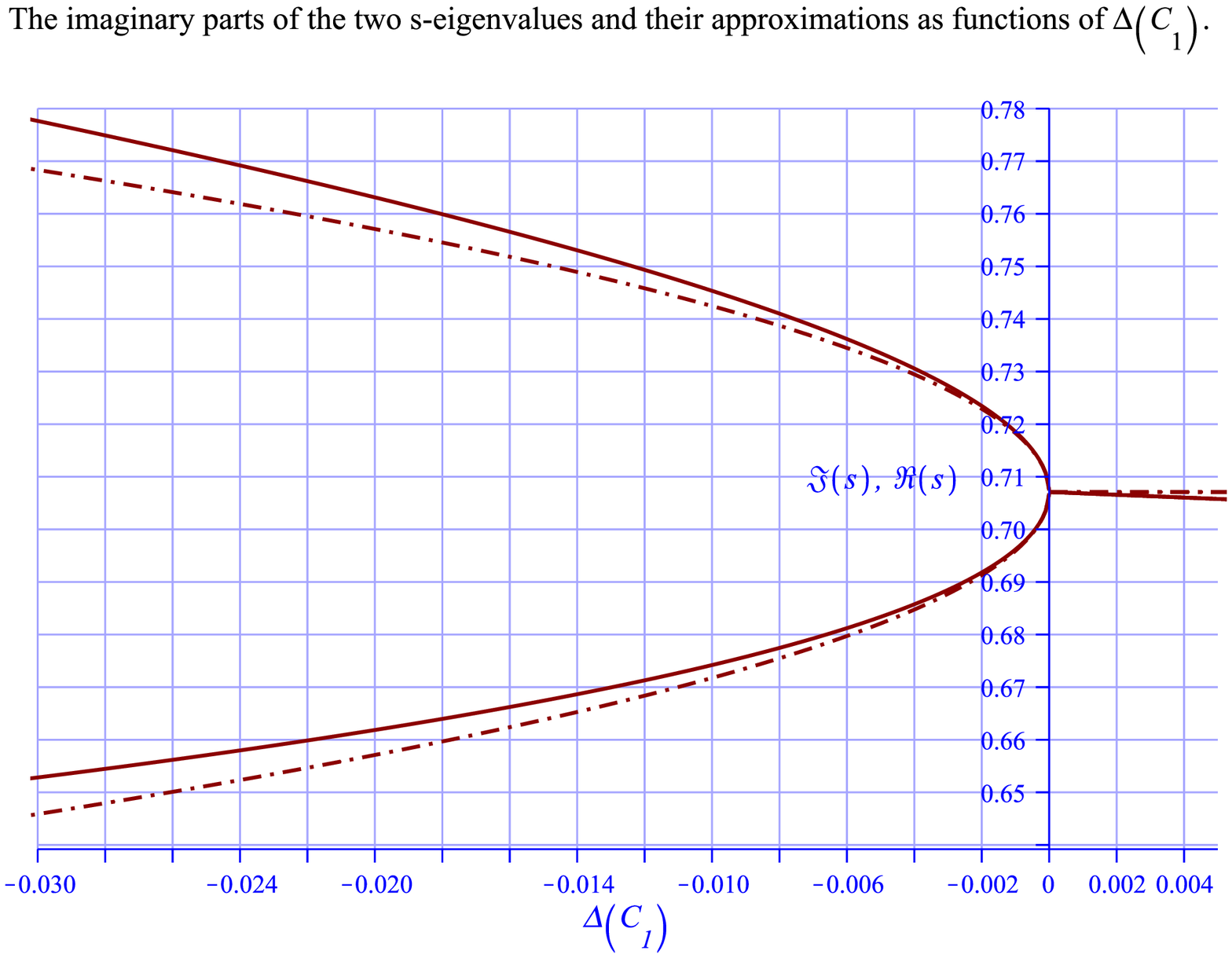}$\hspace{0.5cm}$\includegraphics[scale=0.4]{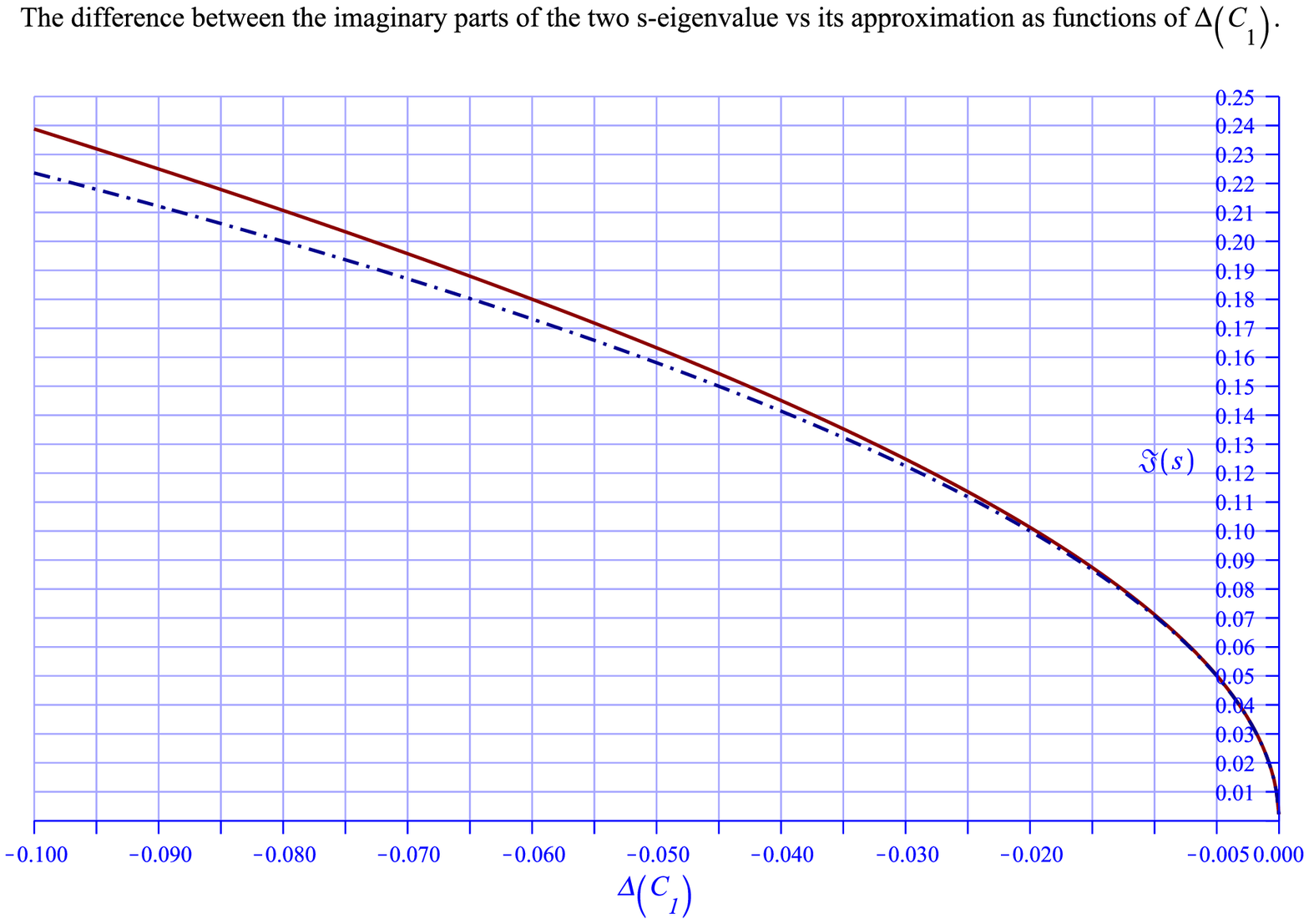}
\par\end{centering}
\centering{}(a)$\hspace{7cm}$(b)\caption{\label{fig:pertC1-2} This plot (a) is fragment of the plot in Figure
\ref{fig:pertC1-1} showing also the imaginary parts of the eigenvalue
approximations (dot-dash line) as in equations (\ref{eq:xixiL3a})
where we assume $\epsilon=1$ and $\Delta\left(R\right)$ to be small.
The plot (b) shows the difference between the imaginary parts of the
eigenvalues as in plot (a) and the corresponding approximation (blue
dash-dot line).}
\end{figure}
\begin{figure}[h]
\centering{}\includegraphics[scale=0.4]{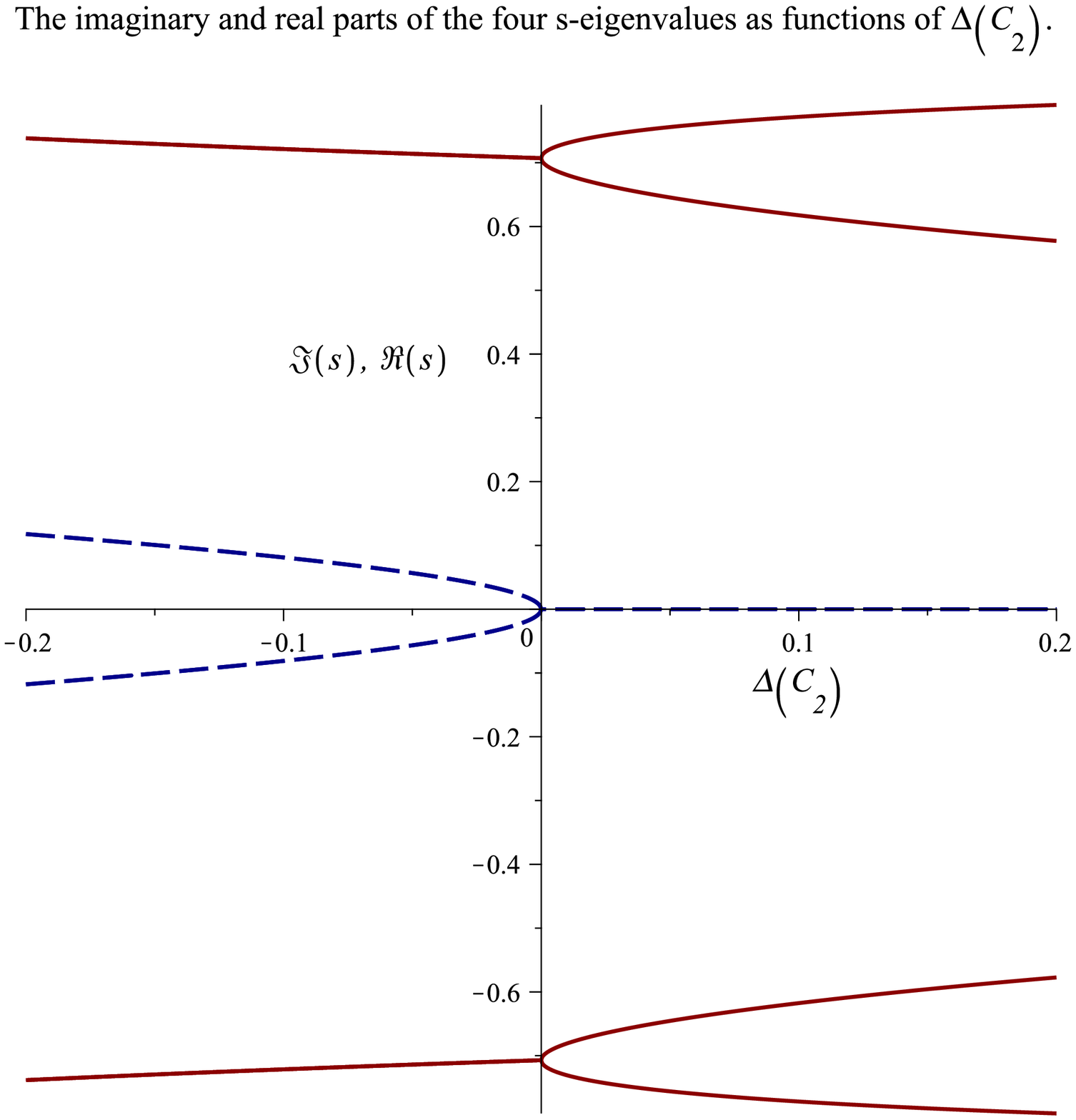}\caption{\label{fig:pertC2-1} This plot shows the real (blue dashed curve)
and imaginary (solid brown curve) the real and imaginary parts of
the four eigenvalues $s$ as in equations (\ref{eq:Adqs2ap}) assuming
that (i) the circuit parameters satisfy equations (\ref{eq:xixiL2b})
where $\epsilon=1$, $\Delta\left(L_{1}\right)=\Delta\left(L_{2}\right)=\Delta\left(C_{1}\right)=\Delta\left(g\right)=0$
and $\Delta\left(C_{2}\right)$ varies. In other words, the plot shows
the variation of the real and imaginary parts of the four eigenvalues
$s$ as function on the capacitance $C_{2}$.}
\end{figure}
\begin{figure}[h]
\begin{centering}
\includegraphics[scale=0.4]{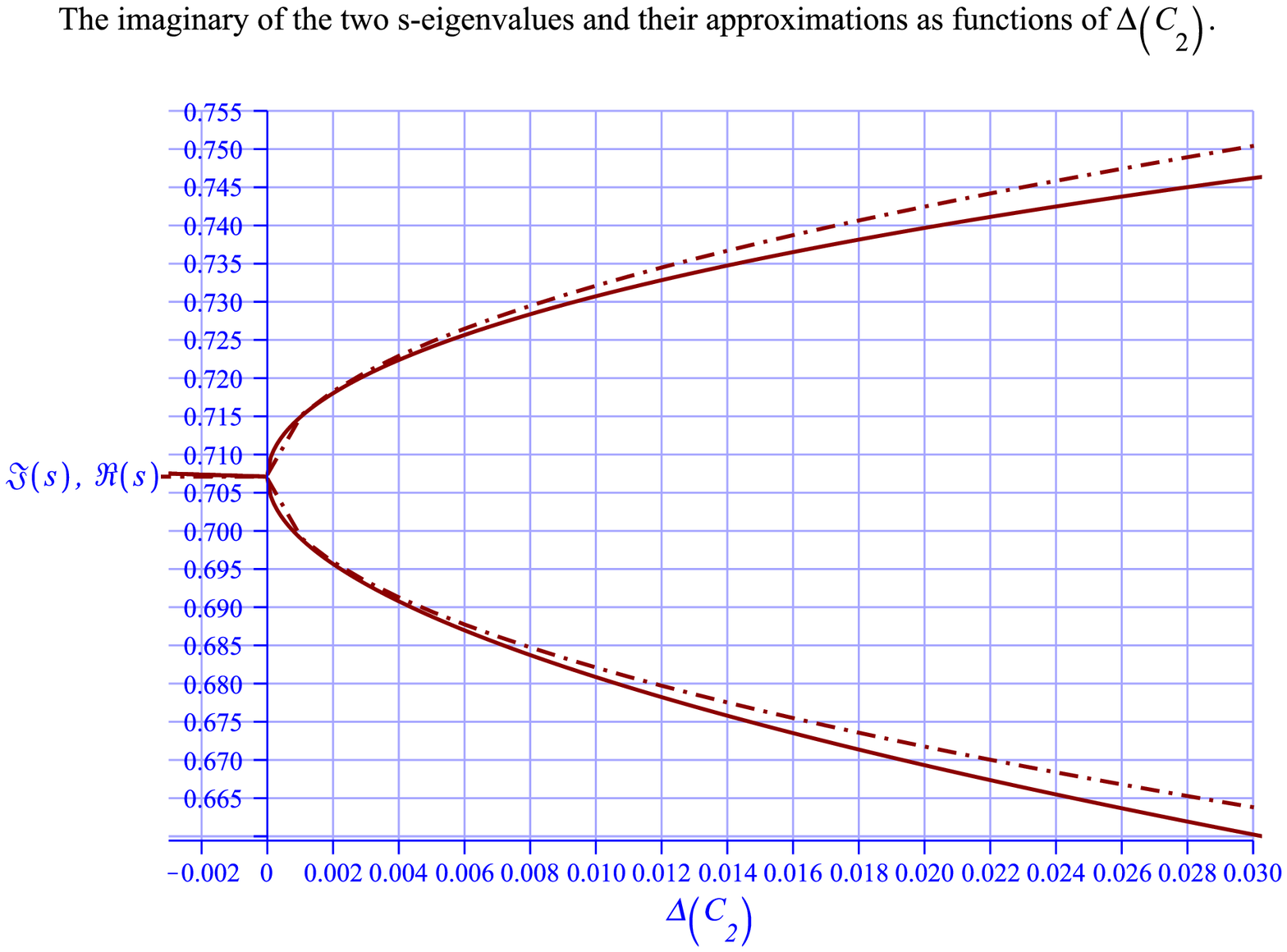}$\hspace{0.5cm}$\includegraphics[scale=0.4]{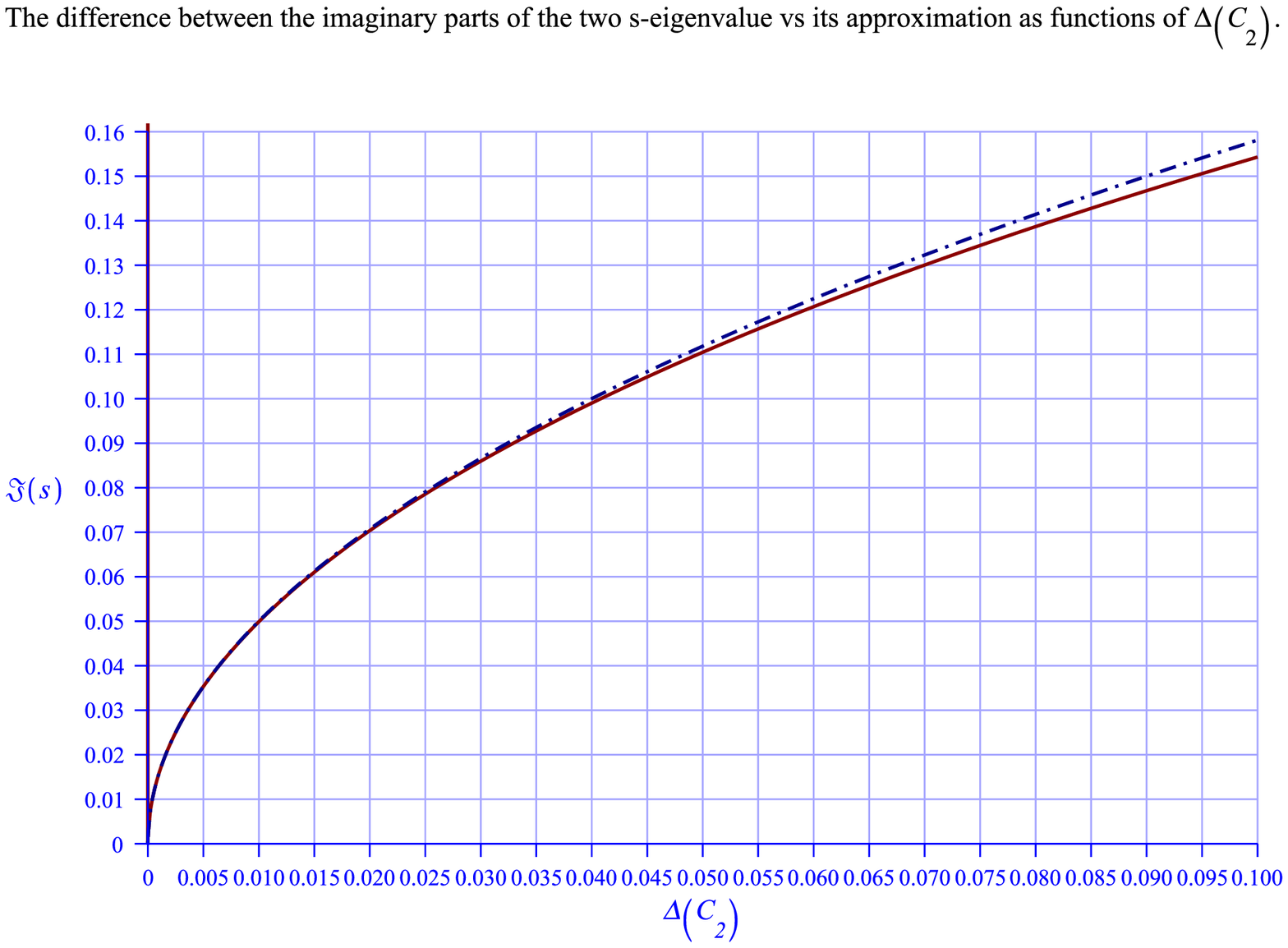}
\par\end{centering}
\centering{}(a)$\hspace{7cm}$(b)\caption{\label{fig:pertC2-2} This plot (a) is fragment for the plot in Figure
\ref{fig:pertC2-1} showing also the imaginary parts of the eigenvalue
approximations (dot-dash line) as in equations (\ref{eq:xixiL3a})
where we assume $\epsilon=1$ and $\Delta\left(R\right)$ to be small.
The plot (b) shows the difference between the imaginary parts of the
eigenvalues as in plot (a) and the corresponding approximation (blue
dash-dot line).}
\end{figure}
\begin{figure}[h]
\centering{}\includegraphics[scale=0.4]{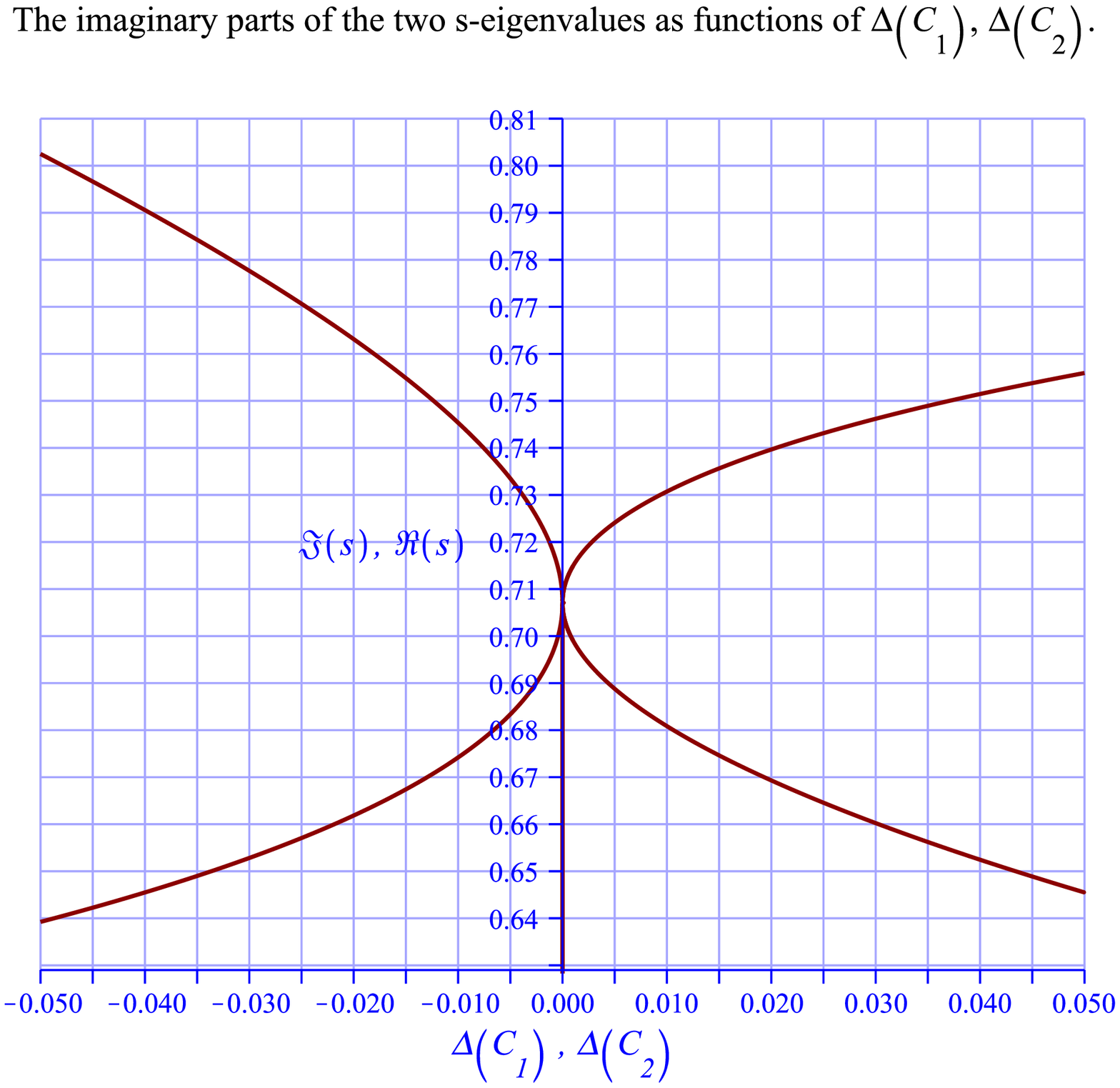}\caption{\label{fig:pertC12} This plot provides for comparative picture of
the variation of the imaginary parts of the eigenvalues shown in Figs.
\ref{fig:pertC1-1} and \ref{fig:pertC2-1} as $\Delta\left(C_{1}\right)$
and $\Delta\left(C_{2}\right)$ vary. The lines on the left and on
the right represent respectively the variation of $s$ as $\Delta\left(C_{1}\right)$
and $\Delta\left(C_{2}\right)$ vary. Notice that the imaginary part
of the eigenvalues vary more with variation of $\Delta\left(C_{1}\right)$
then with variation of $\Delta\left(C_{2}\right)$ in compliance with
equations (\ref{eq:xixiL5a}) and (\ref{eq:xixiL5c}) and Theorem
\ref{thm:omDelR}.}
\end{figure}
One can see in Figure \ref{fig:pertg-1} that there are exactly two
special values of the gyrator parameter $g$, namely $g_{_{-1}}=\frac{1}{2}$
and $g_{_{1}}=\frac{9}{2}$ corresponding to respectively to $\Delta\left(g\right)=0$
and $\Delta\left(g\right)=8$, for which the eigenfrequencies degenerate,
see also Fig. Fig. \ref{fig:eigenvals}.
\begin{figure}[h]
\centering{}\includegraphics[scale=0.4]{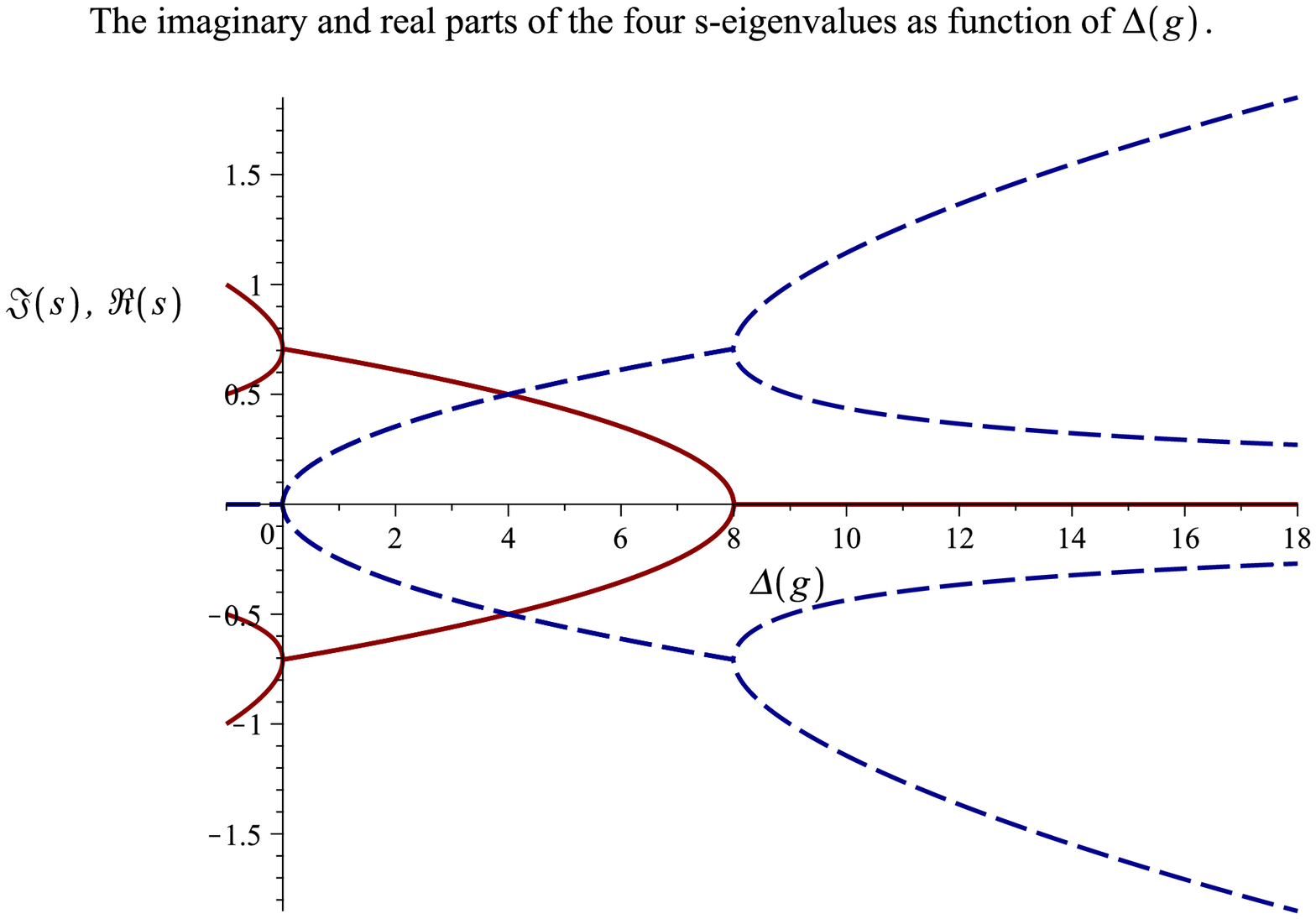}\caption{\label{fig:pertg-1} This plot shows the real (blue dashed curve)
and imaginary (solid brown curve) the real and imaginary parts of
the four eigenvalues $s$ as in equations (\ref{eq:Adqs2ap}) assuming
that (i) the circuit parameters satisfy equations (\ref{eq:xixiL2b})
where $\epsilon=1$, $\Delta\left(L_{1}\right)=\Delta\left(L_{2}\right)=\Delta\left(C_{1}\right)=\Delta\left(C_{2}\right)=0$
and $\Delta\left(g\right)$ varies. In other words, the plot shows
the variation of the real and imaginary parts of the four eigenvalues
$s$ as function on the capacitance $g$.}
\end{figure}
\begin{figure}[h]
\begin{centering}
\includegraphics[scale=0.35]{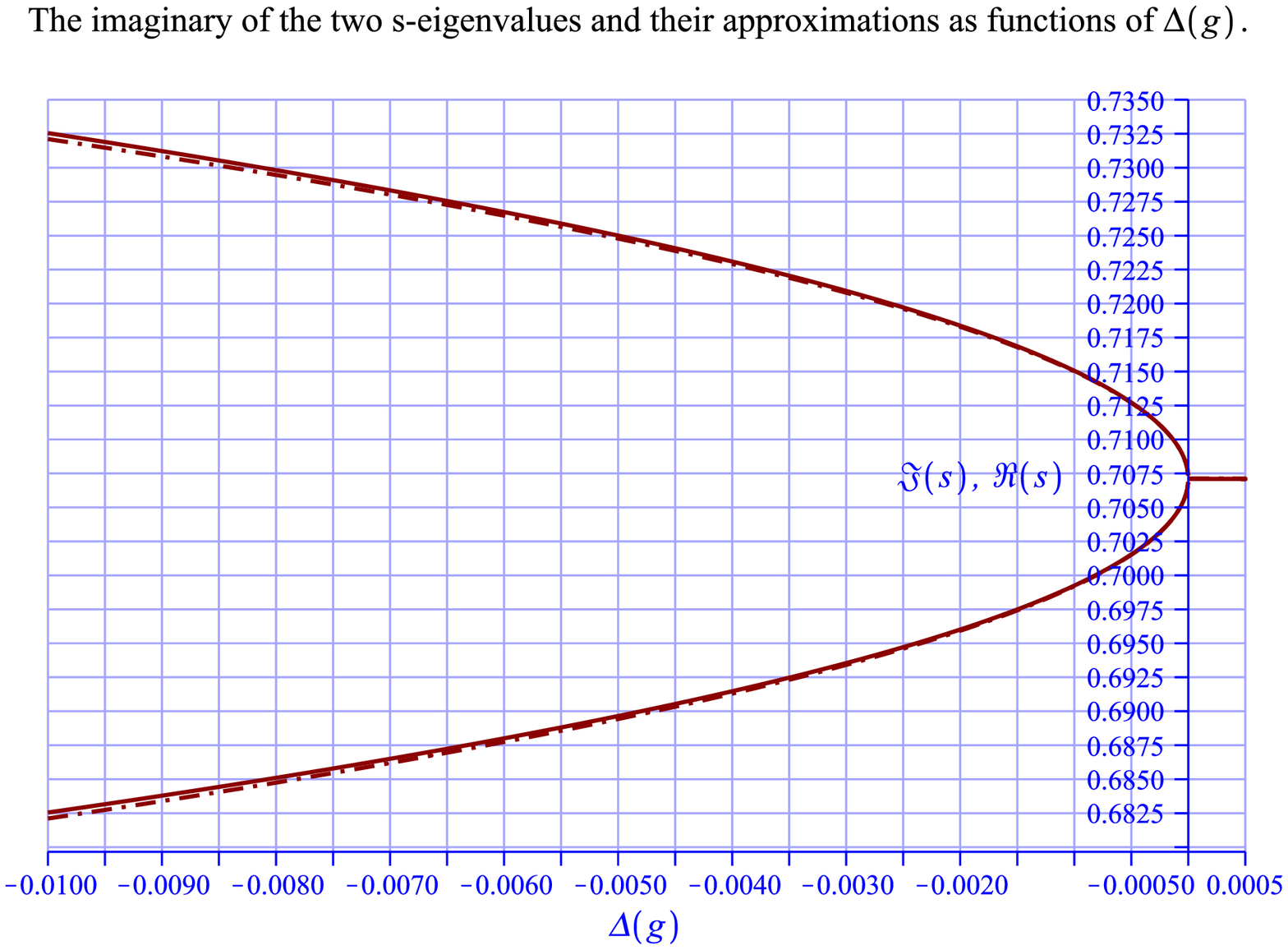}$\hspace{0.5cm}$\includegraphics[scale=0.35]{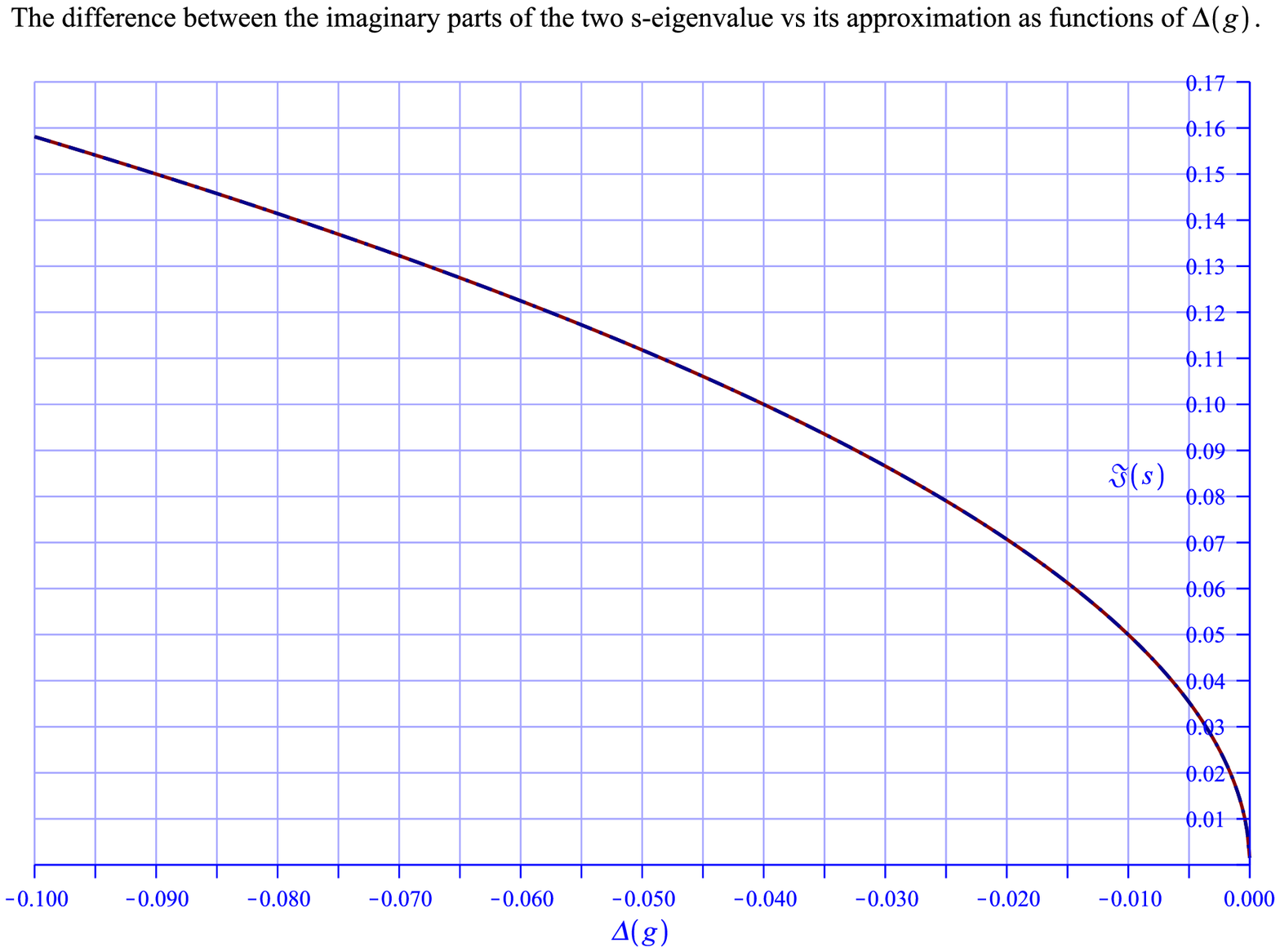}
\par\end{centering}
\centering{}(a)$\hspace{7cm}$(b)\caption{\label{fig:pertg-2} This plot (a) is fragment for the plot in Figure
\ref{fig:pertg-1} showing also the imaginary parts of the eigenvalue
approximations (dot-dash line) as in equations (\ref{eq:xixiL3a})
where we assume $\epsilon=1$ and $\Delta\left(R\right)$ to be small.
The plot (b) shows the difference between the imaginary parts of the
eigenvalues as in plot (a) and the corresponding approximation (blue
dash-dot line). Evidently, the approximation is essentially the same
as the difference between the imaginary parts of the eigenvalues.}
\end{figure}

\section{The circuit stable operation when close to an EPD}

The circuit stable operation is critical to its eigenfrequencies measurements.
Using an EPD for enhanced sensitivity poses an immediate challenge
when it comes to the stability issue. By the very nature of the eigenfrequency
degeneracy the circuit when close to an EPD is only marginally stable
at the best. A clear manifestation of the instability is the existence
of the circuit eigenvalues (eigenfrequencies) $s=\mathrm{i}\omega$
that have non-zero real part with consequent exponential growth in
time when the circuit is in the corresponding eigenstate. Figs. \ref{fig:eigenvals}
and \ref{fig:pertC1-1} provide for graphical illustration demonstrating
the later point. Relations (\ref{eq:xixiL5a}) and (\ref{eq:xixiL5c})
make the point on the instability analytically. Indeed, for the data
as in equations (\ref{eq:xixiL4a}) and (\ref{eq:xixi4Lb}) relations
(\ref{eq:xixiL5a}) imply that if $\frac{C_{1}-\dot{C}_{1}}{\dot{C}_{1}}<0$
then the relative frequency shift $\frac{\omega_{+}-\omega_{-}}{\dot{\omega}}$
is real valued and the circuit is stable. But if $\frac{C_{1}-\dot{C}_{1}}{\dot{C}_{1}}>0$
then the relative frequency shift $\frac{\omega_{+}-\omega_{-}}{\dot{\omega}}$
has non-zero imaginary part and consequently the circuit is unstable.
This analysis suggest a trade-off approach for combining advantages
of employing an EPD point of the circuit for the enhanced sensitivity
and the circuit stability. \emph{This trade-off approach is to use
as the circuit work point its state close to an EPD but not exactly
at it so all the circuit eigenfrequencies are real and consequently
its operation is stable. The close proximity to an EPD point should
be chosen so that the expected circuit perturbations should reliably
maintain the circuit stability}.

To quantify the suggested above trade-off between the circuit enhanced
sensitivity and its stability we proceed as follows. Suppose that
quantity that varies about its EPD value is the capacitance $C_{1}$
with understanding that other circuit parameters variation can be
treated similarly. Having in mind the relative frequency shift $\frac{\omega_{+}-\omega_{-}}{\dot{\omega}}$
defined by equation (\ref{eq:xixiL3ab}) let consider its approximation
(\ref{eq:xixiL5a}) for capacitance $C_{1}$ varying about its EPD
value $\dot{C}_{1}$
\begin{equation}
\varDelta\left(c\right)=\sqrt{-ac}\cong\frac{\omega_{+}-\omega_{-}}{\dot{\omega}},\quad a=\frac{\omega_{1}}{\omega_{2}}-1,\quad c=\Delta\left(C_{1}\right),\quad\left|c\right|\ll1.\label{eq:Delac1a}
\end{equation}
Notice that if $ac>0$ then $\varDelta\left(c\right)$ is pure imaginary
manifesting the circuit instability. Since we want to assure the circuit
stable operation we need to enforce inequality $ac<0$. We achieve
that by partitioning $c$ into two numbers
\begin{equation}
c=w+x,\quad aw<0,\quad\left|x\right|\leq\frac{\left|w\right|}{2},\quad\left|w\right|\ll1.\label{eq:Delac1b}
\end{equation}
Number $w$ in the partitioning (\ref{eq:Delac1b}) is assumed to
be chosen and fixed whereas number $x$ can vary within allowed limit
$\left|x\right|\leq\frac{\left|w\right|}{2}$ implying inequality
$ac<0$ and consequently the stability. We refer to number $w$ as
the \emph{work point}. We introduce also the variation function about
the work point
\begin{equation}
\varDelta_{w}\left(x\right)=\varDelta\left(w+x\right)-\varDelta\left(w\right)=\sqrt{-a\left(w+x\right)}-\sqrt{-aw},\label{eq:Delac1c}
\end{equation}
and the \emph{enhancement factor function}
\begin{equation}
F_{w}\left(x\right)=\left|\frac{\varDelta_{w}\left(x\right)}{x}\right|=\left|\frac{\sqrt{-a\left(w+x\right)}-\sqrt{-aw}}{x}\right|.\label{eq:Delac1d}
\end{equation}
Notice that the following asymptotic formulas holds for variation
function $\varDelta_{w}\left(x\right)$ defined by equation (\ref{eq:Delac1c})
\begin{equation}
\varDelta_{w}\left(x\right)=\frac{\sqrt{-aw}}{2w}x-\frac{\sqrt{-aw}}{8w^{2}}x^{2}+\frac{\sqrt{-aw}}{16w^{3}}x^{3}+O\left(x^{4}\right),\quad x\rightarrow0.\label{eq:Delac1e}
\end{equation}
Asymptotic expression (\ref{eq:Delac1e}) in view of equation (\ref{eq:Delac1b})
readily implies the following representation for the enhancement factor
function $F_{w}\left(x\right)$
\begin{equation}
F_{w}\left(x\right)=\left|\frac{\sqrt{-aw}}{2w}-\frac{\sqrt{-aw}}{8w^{2}}x+\frac{\sqrt{-aw}}{16w^{3}}x^{2}+O\left(x^{3}\right)\right|,\quad x\rightarrow0.\label{eq:Delac1f}
\end{equation}

\begin{rem}[quantifying trade-off]
\label{rem:tradeoff} The expression (\ref{eq:Delac1e}) for enhancement
factor $F_{w}\left(0\right)$ and relations (\ref{eq:Delac1b}) quantify
the trade-off for taking smaller values of $w.$ Indeed smaller $w$
yield larger enhancement factor but the stability condition $\left|x\right|\leq\frac{\left|w\right|}{2}$
requires constraints the allowed variation of $x$ to a smaller interval.
The indicated trade-off allows to identify the smallest $w$ acceptable
for desired range of values for $x$. 
\end{rem}

Figure \ref{fig:enhfac1} shows the values of the enhancement factor
function $F_{w}\left(x\right)$ for different values of $w$ and $a=\frac{\omega_{1}}{\omega_{2}}-1$.

\begin{figure}[h]
\begin{centering}
\includegraphics[scale=0.4]{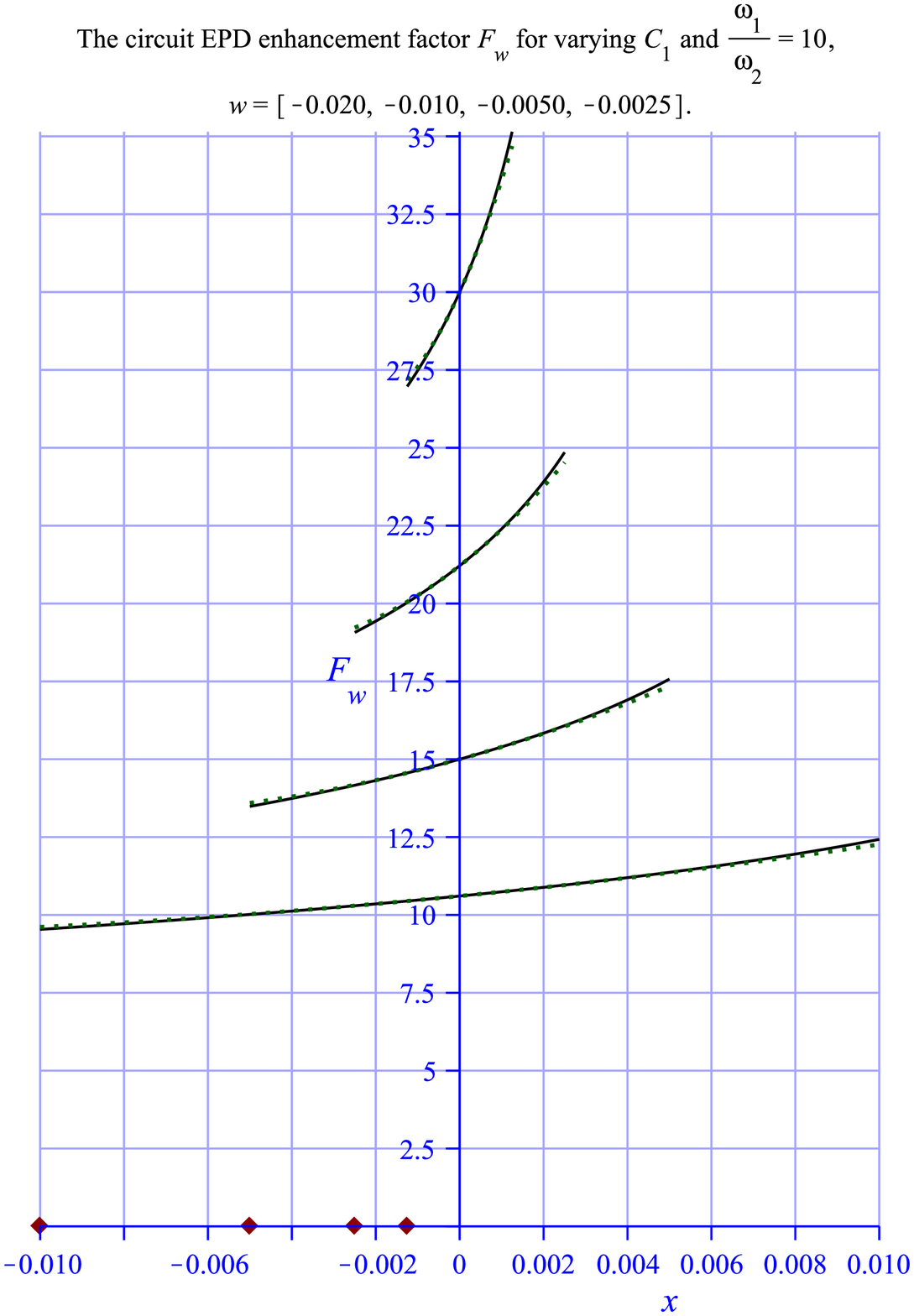}$\hspace{0.5cm}$\includegraphics[scale=0.4]{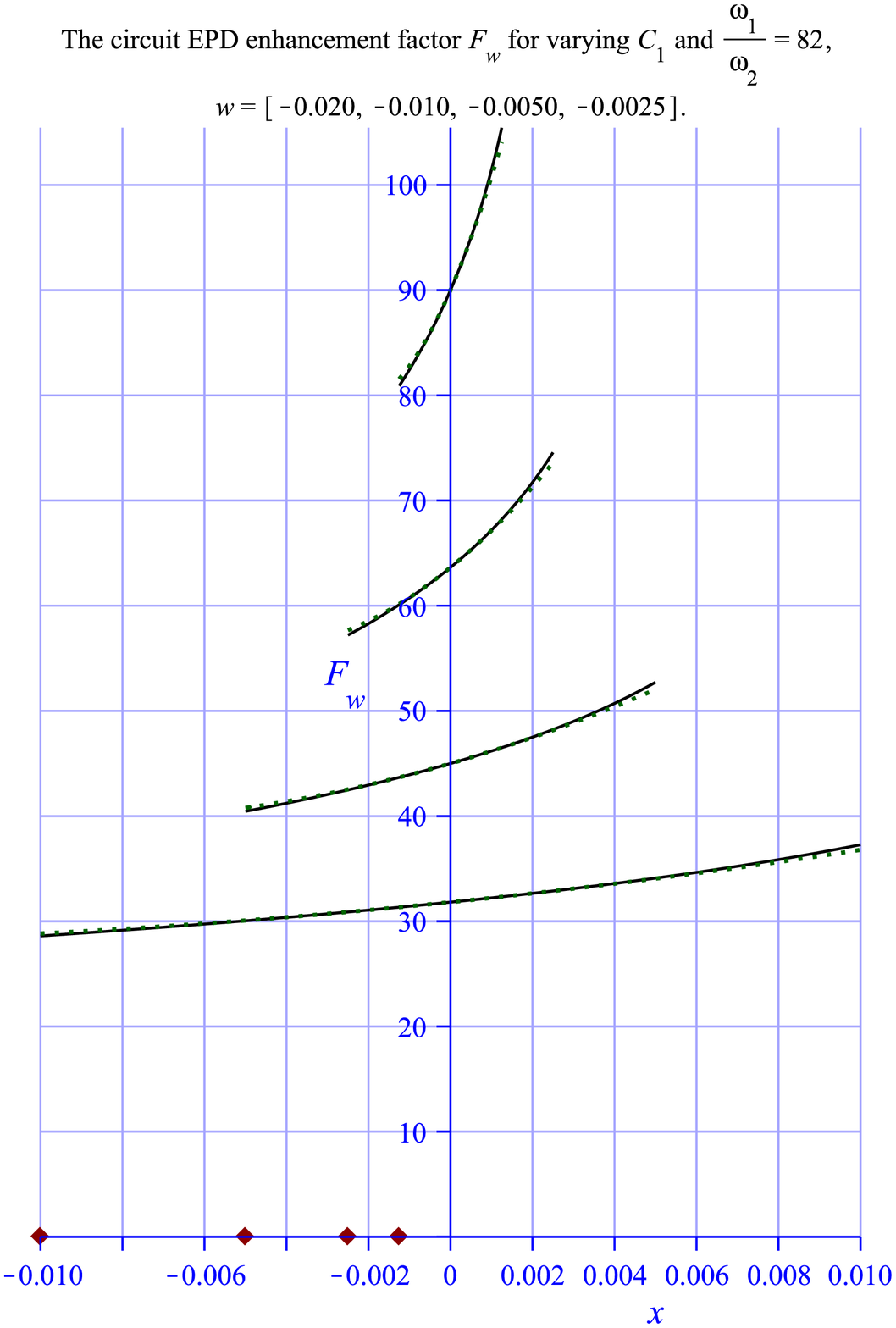}
\par\end{centering}
\centering{}(a)$\hspace{7cm}$(b)\caption{\label{fig:enhfac1} This plot shows the enhancement factor faction
$F_{w}\left(x\right)$ defined in equation (\ref{eq:Delac1e}) as
solid (black) lines and its approximation as in equation (\ref{eq:Delac1f})
as dotted (green) lines for (i) a number of values $w=-0.02,-0.01,-0.005,-0.0025$;
(ii) two values $\frac{\omega_{1}}{\omega_{2}}=10$ for plot (a) and
$\frac{\omega_{1}}{\omega_{2}}=82$ for the plot (b). The variable
$x$ varies in interval $\left[-\left|w\right|,\left|w\right|\right]$
and the corresponding values of $w$ are identified on $x$-axis by
diamond (brown) dots. Smaller values of $w$ yield larger values of
function $F_{w}\left(x\right)$. Notice also that the function $F_{w}\left(x\right)$
approximations fit extremely well its values over indicated intervals.}
\end{figure}

\section{A few points on sensitivity of measurements}

Since we pursue here enhanced sensitivity and higher resolution measurements
using the circuits described in previous sections we provide a brief
review of sensors and their capabilities based on \cite[5.3]{Bish}
and references there. Microsensors seems to be the class of sensors
particular suited for enhanced sensitivity and higher resolution measurements.
According to \cite[5.3]{Bish} they are typically based on either
measurement of mechanical strain, measurement of mechanical displacement,
or on frequency measurement of a structural resonance. Displacement
measurement are often based on the capacitance readout \cite{BoxGre},
\cite{Brook}, \cite{LeuRud}, \cite{Munt}, \cite{Seid}. Resonant-type
sensors measure frequency and they are generally less susceptible
to noise and thus typically provide a higher resolution measurement.
These type of sensors provide as much as one hundred times the resolution
of analog sensors, though they and more complex and are typically
more difficult to fabricate. Using proposed here circuits can enhance
the analog sensor sensitivity and resolution by more 100 times making
them a competitive alternative to more complex resonant-type sensors.

The reported accuracy of sensors based on (differential) capacitance
is 500 ppm (0.05\%) with the smallest change in sensed capacitance
being about $20\,\mathrm{aF}$, \cite{BoxGre}. As to capacitance
tolerances the tolerance values between 1\% and 20\% are common, and
precision capacitor tolerances range from 0.1\% to 0.5\%, \cite[8.2.3]{WhitH},
\cite[2.6.7]{WhitE}.

Using frequencies for assessing the circuit related quantities is
justified by the fact that time interval and frequency can be measured
with less uncertainty and more resolution than any other physical
quantity. The accuracy of the low cost frequency counters can be $1\times10^{-7}$
in 1 s, \cite[10]{Bish}, \cite[27.8]{Boye}. The current limit for
TIC resolution is about 20 ps, which means that a frequency change
of $2\times10^{-11}$ can be detected in 1 s \cite[10]{Bish}. Today,
the best time and frequency standards can realize the SI second with
uncertainties of $\cong1\times10^{-15}$. 

We remind also that according to the basics of measurement theory
the true value of a measured quantity is known only in the case of
calibration of measurement instruments \cite[1.5]{Rab}. In this case,
the true value is the value of the measurement standard used in the
calibration, whose inaccuracy must be negligible compared with the
inaccuracy of the measurement instrument being calibrated. If $x$
is true value of measured quantity and $\tilde{x}$ is the result
of measurement then the absolute error $\varDelta x=\tilde{x}-x$,
and the relative error is $\delta x=\frac{\tilde{x}-x}{x}$, \cite[1.5]{Rab}.

When equations (\ref{eq:Delac1d}) and (\ref{eq:Delac1f}) define
the enhancement function $F_{w}\left(x\right)$ and Figure \ref{fig:enhfac1}
shows its values for a range of values of $x$ we would like to introduce
a simpler enhancement characterization as its value at the work point,
that is for $x=0$:
\begin{equation}
F_{w}\left(0\right)=\frac{\sqrt{\left|a\right|}}{2\sqrt{\left|w\right|}}=\frac{\sqrt{\left|\frac{\omega_{1}}{\omega_{2}}-1\right|}}{2\sqrt{\left|w\right|}}.\label{eq:Delac1g}
\end{equation}
It turns out that quantity $F_{w}\left(0\right)$ is relevant to the
resolution of measurements at the work point. Indeed, equation (\ref{eq:Delac1a})
implies that the resolution $\rho_{x}$ for $x$ at $x=0$, that is
at the work point $w$, is related to the resolution $\rho_{f}$ for
the relative frequency $\frac{\omega_{+}-\omega_{-}}{\dot{\omega}}$
by the following equation
\begin{equation}
\rho_{f}=\left|\partial_{c}\varDelta\left(w\right)\right|\rho_{x}.\label{eq:Delac2b}
\end{equation}
In view of equations (\ref{eq:Delac1a}), (\ref{eq:Delac1c}), (\ref{eq:Delac1f})
and (\ref{eq:Delac1g}) the factor $\left|\partial_{c}\varDelta\left(w\right)\right|$
in the right-hand side of equation (\ref{eq:Delac2b}) satisfies
\begin{equation}
\left|\partial_{c}\varDelta\left(w\right)\right|=\left|\partial_{x}\varDelta_{w}\left(0\right)\right|=\left|\frac{\sqrt{-aw}}{2w}\right|=F_{w}\left(0\right)=\frac{\sqrt{\left|\frac{\omega_{1}}{\omega_{2}}-1\right|}}{2\sqrt{\left|w\right|}}.\label{eq:Delac2c}
\end{equation}
Substitution of expression (\ref{eq:Delac2b}) for $\left|\partial_{c}\varDelta\left(w\right)\right|$
into equation equation (\ref{eq:Delac2b}) yields
\begin{equation}
\rho_{x}=\frac{\rho_{f}}{F_{w}\left(0\right)}=\frac{2\sqrt{\left|w\right|}}{\sqrt{\left|\frac{\omega_{1}}{\omega_{2}}-1\right|}}\rho_{f}.\label{eq:Delac2d}
\end{equation}
Figure \ref{fig:enhfac2} illustrates the dependence of the enhancement
factor $F_{w}\left(0\right)$ on $w$ and the circuit parameters.
\begin{figure}[h]
\begin{centering}
\includegraphics[scale=0.4]{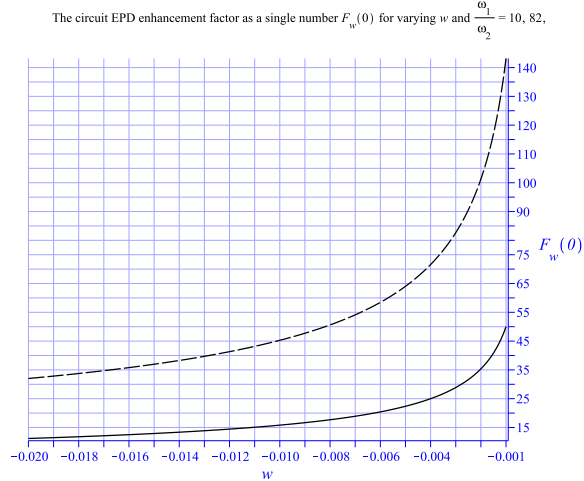}
\par\end{centering}
\centering{}\caption{\label{fig:enhfac2} This plot shows the enhancement factor $F_{w}\left(0\right)$
defined in equation (\ref{eq:Delac1g}) for a range of values of $w$
and for (a) $\frac{\omega_{1}}{\omega_{2}}=10$ as solid line; (b)
$\frac{\omega_{1}}{\omega_{2}}=82$ as dashed line. The shown values
of $F_{w}\left(0\right)$ together with formula (\ref{eq:Delac2d})
relating resolutions $\rho_{x}$ and $\rho_{f}$ shows the circuit
can enhance the resolution $\rho_{x}$ by one or two orders of magnitude.}
\end{figure}

\section{Jordan blocks viewed through their perturbations\label{sec:Jordper}}

The determination by numerical evaluations of the presence of nontrivial
Jordan blocks in the Jordan canonical form of a matrix is a rather
complicated matter. The reason for the problem is that an infinitesimally
small perturbation of the matrix can make it completely diagonalizable.
An explanation to this seemingly paradoxical situation is that though
the perturbed matrix is diagonal some of its relevant eigenvectors
are nearly parallel indicating the perturbed matrix proximity to a
matrix with nontrivial (non-diagonal) Jordan canonical form. The explanation
suggests that a numerically sound determination of the presence of
nontrivial Jordan blocks in the Jordan form of a matrix can be based
on the matrix perturbations and we pursue this idea here.

The perturbation theory of the matrix Jordan structure in the case
of generic perturbations is known as \emph{Lidskii's theory} \cite{Lids},
\cite{MoBuOv}, \cite[7.4]{Baum}, \cite[2.5]{SeyMai}, \cite[4.4.2]{KirO}.
Lidskii's theory allows to determine effectively the first significant
perturbation for the eigenvalues and the Jordan basis vectors. Under
additional conditions the theory yields the power series representations
for the eigenvalues and the Jordan basis vectors. In the case of nonderogatory
(geometric multiplicity one) eigenvalue and a generic perturbation
there is constructive recursive algorithm for computing all coefficients
for these power series due to Welters \cite{Welt}.

Our goal here to develop based on some concepts and results of Lidskii's
theory constructive approaches for detecting nontrivial Jordan blocks
in the Jordan canonical form of a matrix by their manifestation in
the matrix generic perturbations.

Let $\mathbb{C}^{n\times n}$ be the set of $n\times n$ matrices
with complex-valued coefficients and let matrix $A\in\mathbb{C}^{n\times n}$.
Consider now a perturbation $A\left(\varepsilon\right)$ of matrix
$A$ of the following form

\begin{equation}
A\left(\varepsilon\right)=A+B\left(\varepsilon\right),\quad B\left(\varepsilon\right)=B_{1}\varepsilon+\sum_{k=2}^{\infty}B_{k}\varepsilon^{k},\quad A,\;B_{k}\in\mathbb{C}^{n\times n},\label{eq:ABeps1a}
\end{equation}
where $\varepsilon$ is a small parameter and the power series is
assumed to converge. Suppose now that $\lambda_{0}$ is a nonderogatory
eigenvalue of the algebraic multiplicity $m$, that is its geometric
multiplicity is exactly $1$. Suppose $\mathscr{J}$ be the Jordan
form of matrix $A$ and $Z$ be $n\times n$ matrix the columns of
which form the Jordan basis matrix $A$. Then
\begin{equation}
\mathscr{J}=Z^{-1}AZ=\left[\begin{array}{rr}
J_{m}\left(\lambda_{0}\right) & 0\\
0 & J
\end{array}\right],\quad J_{m}\left(\zeta\right)=\left[\begin{array}{ccccc}
\zeta & 1 & \cdots & 0 & 0\\
0 & \zeta & 1 & \cdots & 0\\
0 & 0 & \ddots & \cdots & \vdots\\
\vdots & \vdots & \ddots & \zeta & 1\\
0 & 0 & \cdots & 0 & \zeta
\end{array}\right],\label{eq:AJzeta1a}
\end{equation}
where $J$ is the part of the Jordan form $\mathscr{J}$ related to
the eigenvalues of $A$ different than $\lambda_{0}$ and $J_{m}\left(\zeta\right)$
is the Jordan block of dimension $m$ associated with eigenvalue $\zeta$.
Using equations (\ref{eq:AJzeta1a}) we introduce the following matrices
represented in block form

\begin{equation}
Z^{-1}AZ=\left[\begin{array}{rr}
J_{m}\left(\lambda_{0}\right) & 0\\
0 & J
\end{array}\right],\quad Z^{-1}B_{1}Z=\left[\begin{array}{rr}
\tilde{B}_{1} & \ast\\
\ast & \ast
\end{array}\right],\label{eq:AJzeta1b}
\end{equation}
where symbol ``$\ast$'' signifies matrix blocks of relevant dimensions.
Let us consider now entries of block $\tilde{B}_{1}\in\mathbb{C}^{m\times m}$
\begin{equation}
\tilde{B}_{1}=\left[\begin{array}{rrrrr}
\ast & \ast & \cdots & \ast & \ast\\
\vdots & \ast & \ast & \cdots & \ast\\
\ast & \vdots & \ddots & \cdots & \vdots\\
b_{m-1,1} & \ast & \ddots & \ast & \ast\\
b_{m,1} & b_{m,2} & \cdots & \ast & \ast
\end{array}\right],\label{eq:AJzeta1c}
\end{equation}
 where entries marked by ``$\ast$'' are of no particular significance.

Let $\mathbb{C}^{n}$ be a set of column-vectors with complex-valued
entries and let $e_{j},\:1\leq j\leq n$ be its standard basis. Our
first statement will be on the geometric multiplicity $\mathrm{gmul}{}_{A}\,\left(\lambda_{0}\right)$
of an eigenvalue $\lambda_{0}$ of matrix $A$. We remind that the
geometric multiplicity $\mathrm{gmul}{}_{A}\,\left(\lambda_{0}\right)$
is defined as the dimension of the eigenspace of matrix $A$ corresponding
to the eigenvalue $\lambda_{0}$. The statement holds under certain
genericity assumption and to formulate it we proceed as follows.

Let $V$ be the eigenspace of matrix $A$ corresponding to the eigenvalue
$\lambda_{0}$ and $W$ be the eigenspace of matrix $A^{\ast}$, which
is adjoint to $A$, corresponding to the eigenvalue $\bar{\lambda}_{0}$
which is complex-conjugate to $\lambda_{0}$.
\begin{equation}
\dim V=\dim W=\mathrm{gmul}{}_{A}\,\left(\lambda_{0}\right)=\mathrm{gmul}{}_{A^{\ast}}\,\left(\lambda_{0}^{\ast}\right)=r.\label{eq:ABeps1b}
\end{equation}
Let 
\begin{equation}
A_{0}=A-\lambda_{0}\mathbb{I},\quad U=\left[u_{1},u_{2},\ldots,u_{m},u_{m+1},\ldots,u_{n}\right],\label{eq:AuU1a}
\end{equation}
where matrix $U\in\mathbb{C}^{n\times n}$ is an orthogonal matrix
formed by its column-vectors
\begin{equation}
u_{j}=Ue_{j}\in\mathbb{C}^{n},\quad1\leq j\leq n.\label{eq:AuU1aa}
\end{equation}
We suppose now that the first $r$ column-vectors $u_{j}$ are chosen
to form an orthonormal basis of the vector space $V$ and consequently
they are eigenvectors of matrix $A$ corresponding to the eigenvalue
$\lambda_{0}$, that is
\begin{equation}
A_{0}u_{j}=\left(A-\lambda_{0}\mathbb{I}\right)u_{j}=0,\quad1\leq j\leq r.\label{eq:AuU1b}
\end{equation}
Notice that the general identity $r+\mathrm{rank}\,\left\{ A\right\} =n$
implies that
\begin{equation}
A_{0}u_{j},\quad r+1\leq j\leq n,\text{ are linearly independent.}\label{eq:AuU1c}
\end{equation}
Let matrix $S\in\mathbb{C}^{n\times n}$ be an invertible matrix such
that
\begin{equation}
SA_{0}u_{j}=e_{j},\quad r+1\leq j\leq n,\label{eq:AuU1d}
\end{equation}
 the existence of which is warranted by relations (\ref{eq:AuU1c}).
In view of (\ref{eq:AuU1aa}), (\ref{eq:AuU1b}) and (\ref{eq:AuU1d})
we have
\begin{equation}
SA_{0}Ue_{j}=0,\quad1\leq j\leq n;\quad SA_{0}Ue_{j}=e_{j}\quad r+1\leq j\leq n,\label{eq:AuU1e}
\end{equation}
 and consequently matrix $SA_{0}U\in\mathbb{C}^{n\times n}$ has the
following block form
\begin{equation}
SA_{0}U=\left[\begin{array}{rr}
0 & 0\\
0 & \mathbb{I}
\end{array}\right],\quad A_{0}=A-\lambda_{0}\mathbb{I},\label{eq:AuU1f}
\end{equation}
where upper diagonal block matrix $0\in\mathbb{C}^{r\times r}$ ,
$\mathbb{I}\in\mathbb{C}^{\left(n-r\right)\times\left(n-r\right)}$
and off-diagonal block matrices are $0$ matrices of the relevant
dimensions.
\begin{thm}[geometric multiplicity test]
 Suppose that $\lambda_{0}$ is an eigenvalue of matrix $A\in\mathbb{C}^{n\times n}$
of the geometric multiplicity
\begin{equation}
r=\mathrm{gmul}{}_{A}\,\left(\lambda_{0}\right).\label{eq:detAB1aa}
\end{equation}
Let matrix $B\left(\varepsilon\right)\in\mathbb{C}^{n\times n}$ be
represented by converging power series
\begin{equation}
B\left(\varepsilon\right)=B_{1}\varepsilon+\sum_{k=2}^{\infty}B_{k}\varepsilon^{k},\quad B_{k}\in\mathbb{C}^{n\times n},\quad1\leq k\leq n.\label{eq:detAB1ab}
\end{equation}
Then
\begin{equation}
\det\left\{ A+B\left(\varepsilon\right)-\lambda_{0}\mathbb{I}\right\} =O\left(\varepsilon^{r}\right).\label{eq:detAB1a}
\end{equation}
Let matrices $U$ and $S$ be as the described before the statement
of this theorem. Suppose that matrix $SB_{1}U$ is partitioned conformally
with matrix $SA_{0}U$ in (\ref{eq:AuU1f}) and is of the following
block form
\begin{equation}
SB_{1}U=\left[\begin{array}{cr}
b_{1} & *\\
* & *
\end{array}\right],\label{eq:detAB1b}
\end{equation}
where $b_{1}\in\mathbb{C}^{r\times r}$. Assume that matrix $SB_{1}U$
satisfies the genericity condition in the form
\begin{equation}
\det\left\{ b_{1}\right\} \neq0.\label{eq:detAB1c}
\end{equation}
Then
\begin{equation}
\det\left\{ A+B\left(\varepsilon\right)-\lambda_{0}\mathbb{I}\right\} =d\varepsilon^{r}+O\left(\varepsilon^{r+1}\right),\quad d=\frac{\det\left\{ b_{1}\right\} }{\det\left\{ S\right\} }\neq0,\quad\varepsilon\rightarrow0.\label{eq:detAB1d}
\end{equation}
\end{thm}

\begin{proof}
Using block s (\ref{eq:AuU1a}).representation (\ref{eq:detAB1b})
and equation (\ref{eq:AuU1f}) we obtain
\begin{equation}
S\left(A+B\left(\varepsilon\right)-\lambda_{0}\mathbb{I}\right)U=S\left(A_{0}+B\left(\varepsilon\right)\right)U=\left[\begin{array}{rr}
b_{1}\varepsilon+O\left(\varepsilon^{2}\right) & O\left(\varepsilon\right)\\
O\left(\varepsilon\right) & \mathbb{I}+O\left(\varepsilon\right)
\end{array}\right],\quad\varepsilon\rightarrow0.\label{eq:detAB1e}
\end{equation}
Then applying equation (\ref{eq:Block1f}) of Proposition \ref{prop:4blockA1}
to the block matrix in the right-hand side of equations (\ref{eq:detAB1e})
\begin{equation}
\det\left\{ S\left(A+B\left(\varepsilon\right)-\lambda_{0}\mathbb{I}\right)U\right\} =\varepsilon^{r}\det\left\{ b_{1}\right\} +O\left(\varepsilon^{r+1}\right),\quad\varepsilon\rightarrow0.\label{eq:detAB1f}
\end{equation}
Since $U$ is an orthogonal matrix $\det U=1$ equation (\ref{eq:detAB1f})
readily implies relations (\ref{eq:detAB1a}) and (\ref{eq:detAB1d}).
\end{proof}
The presence of a nontrivial Jordan block in the Jordan form of a
matrix $A$ can be detected in a number of ways based on its perturbation
$A\left(\varepsilon\right)=A+B\left(\varepsilon\right)$. The following
statement due to Welters, \cite{Welt}, establishes the equivalency
of some of those ways.
\begin{prop}[perturbation of nonderogatory eigenvalue]
\label{prop:nonder} Let $A\left(\varepsilon\right)$ be $\mathbb{C}^{n\times n}$
matrix-valued analytic function of $\varepsilon$ for small $\varepsilon$.
Let $\lambda_{0}$ be an eigenvalue of the unperturbed matrix $A\left(0\right)$
of the algebraic multiplicity $m$. Then the following statements
are equivalent:
\begin{enumerate}
\item The characteristic polynomial $\det\left\{ \lambda\mathbb{I}-A\left(\varepsilon\right)\right\} $
has a simple zero with respect to $\varepsilon$ at $\lambda=\lambda_{0}$
and $\varepsilon=0$, that is
\begin{equation}
\frac{\partial}{\partial\varepsilon}\left.\det\left\{ \lambda\mathbb{I}-A\left(\varepsilon\right)\right\} \right|_{\left(\varepsilon,\lambda\right)=\left(0,\lambda_{0}\right)}\neq0.\label{eq:wel1a}
\end{equation}
\item The characteristic equation $\det\left\{ \lambda\mathbb{I}-A\left(\varepsilon\right)\right\} =0$
has a unique solution, $\varepsilon\left(\lambda\right)$, in a neighborhood
of $\lambda=\lambda_{0}$ with $\varepsilon\left(\lambda_{0}\right)=0$.
This solution is an analytic function with a zero of order $m$ at
$\lambda=\lambda_{0}$, i.e.,
\begin{equation}
\varepsilon\left(\lambda_{0}\right)=\left.\frac{d\varepsilon\left(\lambda\right)}{d\lambda}\right|_{\lambda=\lambda_{0}}=\cdots=\left.\frac{d^{m-1}\varepsilon\left(\lambda\right)}{d\lambda^{m-1}}\right|_{\lambda=\lambda_{0}}=0,\quad\left.\frac{d^{m}\varepsilon\left(\lambda\right)}{d\lambda^{m}}\right|_{\lambda=\lambda_{0}}\neq0.\label{eq:wel1b}
\end{equation}
\item There exists a convergent Puiseux series whose branches are given
by
\begin{equation}
\lambda_{h}\left(\varepsilon\right)=\lambda_{0}+\alpha_{1}\zeta^{h}\varepsilon^{\frac{1}{m}}+\sum_{k=2}^{\infty}\alpha_{k}\left(\zeta^{h}\varepsilon^{\frac{1}{m}}\right)^{k},\quad h=0,\ldots,m-1,\quad\zeta=\mathrm{e}^{\frac{2\pi}{m}\mathrm{i}},\label{eq:wel1c}
\end{equation}
for any fixed branch of $\varepsilon^{\frac{1}{m}}$, and the first
order term is nonzero, i.e.,
\begin{equation}
\alpha_{1}\neq0.\label{eq:wel1d}
\end{equation}
The values of the branches $\lambda_{h}\left(\varepsilon\right)$
give all the solutions of the characteristic equation for sufficiently
small $\varepsilon$ and $\lambda$ sufficiently near $\lambda_{0}$.
\item The Jordan normal form of $A\left(0\right)$ corresponding to the
eigenvalue $\lambda_{0}$ consists of a single $m\times m$ Jordan
block, and there exists an eigenvector $u_{0}$ of $A\left(0\right)$
corresponding to the eigenvalue $\lambda_{0}$ and an eigenvector
$v_{0}$ of matrix $A^{\ast}\left(0\right)$ (adjoint to $A\left(0\right)$
) corresponding to the eigenvalue $\bar{\lambda}_{0}$ such that
\begin{equation}
\left(v_{0},B_{1}u_{0}\right)\neq0,\quad B_{1}=\left.\frac{dA\left(\varepsilon\right)}{\partial\varepsilon}\right|_{\varepsilon=0}.\label{eq:wel1e}
\end{equation}
\end{enumerate}
\end{prop}

One of the ways to utilize Proposition \ref{prop:nonder} for detecting
a Jordan block in the Jordan form of a matrix $A$ is to establish
numerically that the perturbed matrix $A+B\left(\varepsilon\right)$
has eigenvalues consistent with the Puiseux series (\ref{eq:wel1c})
involving fractional power $\varepsilon^{\frac{1}{m}}$. That would
allow to claim that the the Jordan normal form of $A\left(0\right)$
corresponding to the eigenvalue $\lambda_{0}$ consists of a single
$m\times m$ Jordan block.

An alternative way to detect the presence of a Jordan block associated
with A is to utilize the statement (ii) of Proposition \ref{prop:nonder}.
Namely, one can try to establish numerically that there is a unique
solution $\varepsilon\left(\lambda\right)$ to characteristic equation
$\det\left\{ \lambda\mathbb{I}-A\left(\varepsilon\right)\right\} =0$
in the vicinity of $\varepsilon=0$ and $\lambda=\lambda_{0}$ which
is consistent with the following equations
\begin{equation}
\frac{\varepsilon\left(\lambda\right)}{\left(\lambda-\lambda_{0}\right)^{m}}=c_{1}+O\left(\left(\lambda-\lambda_{0}\right)\right),\!\quad c>0,\!\quad\lambda\rightarrow\lambda_{0},\label{eq:wel1f}
\end{equation}
where $m$ is a positive integer. If the relation (\ref{eq:wel1f})
for $\varepsilon\left(\lambda\right)$ is established it would imply
that the the Jordan normal form of $A\left(0\right)$ corresponding
to the eigenvalue $\lambda_{0}$ consists of a single $m\times m$
Jordan block.

There is yet another useful statement for verifying that the Jordan
form of matrix $A$ has a nontrivial Jordan block. To formulate it
we introduce the following operator norm on the set of matrices $\mathbb{C}^{n\times n}$:
\begin{equation}
\left\Vert A\right\Vert _{\infty}=\max_{1\leq j\leq n}\left\{ \sum_{j=1}^{n}\left|a_{ij}\right|\right\} ,\quad A=\left\{ a_{ij}\right\} _{i,j=1}^{n},\label{eq:Ainf1a}
\end{equation}
where $i$ and $j$ are respectively the row the column indexes. Here
is the statement due Bauer and Fike \cite[6.3]{HorJohn}.
\begin{prop}[perturbation of a diagonalizable matrix]
\label{prop:perdiag} Let $A$ be $\mathbb{C}^{n\times n}$ matrix
which is diagonalizable, that is $A=S\varLambda S^{-1},$ in which
$S$ is nonsingular and $\varLambda$ is diagonal. Let $B$ be $\mathbb{C}^{n\times n}$matrix.
If $\tilde{\lambda}$ is an eigenvalue of $A+B$, there is an eigenvalue
$\lambda$ of $A$ such that
\begin{equation}
\left|\tilde{\lambda}-\lambda\right|\leq\kappa_{\infty}\left(S\right)\left\Vert B\right\Vert _{\infty},\quad\kappa_{\infty}\left(S\right)=\left\Vert S\right\Vert _{\infty}\left\Vert S^{-1}\right\Vert _{\infty},\label{eq:Ainf1b}
\end{equation}
in which $\kappa_{\infty}\left(S\right)$ is the co-called condition
number with respect to the norm $\left\Vert \ast\right\Vert _{\infty}$.
\end{prop}

A way to utilize Proposition \ref{prop:perdiag} is to consider perturbed
matrix $A+B\left(\varepsilon\right)$ and establish numerically that
$\left|\tilde{\lambda}-\lambda\right|>c\varepsilon^{\alpha}$ where
$c>0$ and $0<\alpha<1$. Then since the indicated inequality involves
fractional power $\varepsilon^{\alpha}$ with $0<\alpha<1$ it would
be inconsistent with inequality (\ref{eq:Ainf1b}) where $\left\Vert B\left(\varepsilon\right)\right\Vert _{\infty}=\left|\varepsilon\right|\left\Vert B\right\Vert _{\infty}$implying
that unperturbed matrix $A$ can not be diagonalized, and consequently
that the Jordan form of $A$ has to have a nontrivial Jordan block.

\section{Appendix}

\subsection{Appendix B: Some properties of block matrices\label{sec:block-mat}}

The statements on block matrices below are useful for our studies
\cite[2.8]{BernM}.
\begin{prop}[factorization of a block matrix]
\label{prop:4blockA1} Let $A\in\mathbb{C}^{n\times n}$, $B\in\mathbb{C}^{n\times m}$,
$C\in\mathbb{C}^{p\times n}$, $D\in\mathbb{C}^{p\times m}$, and
assume $A$ is nonsingular. Then
\begin{gather}
\left[\begin{array}{cr}
A & B\\
C & D
\end{array}\right]=\left[\begin{array}{rc}
I & 0\\
CA^{-1} & I
\end{array}\right]\left[\begin{array}{cr}
A & 0\\
0 & D-CA^{-1}B
\end{array}\right]\left[\begin{array}{rr}
I & A^{-1}B\\
0 & I
\end{array}\right],\label{eq:BLock1d}
\end{gather}
and
\begin{gather}
\mathrm{Rank}\,\left\{ \left[\begin{array}{cr}
A & B\\
C & D
\end{array}\right]\right\} =n+\mathrm{Rank}\,\left\{ D-CA^{-1}B\right\} .\label{eq:BLock1e}
\end{gather}
If furthermore, $m=p$, that is $A\in\mathbb{C}^{n\times n}$, $B\in\mathbb{C}^{n\times m}$,
$C\in\mathbb{C}^{m\times n}$, $D\in\mathbb{C}^{m\times m}$, then\cite[26]{FukSha}
\begin{gather}
\det\left\{ \left[\begin{array}{cr}
A & B\\
C & D
\end{array}\right]\right\} =\det\left\{ A\right\} \det\left\{ D-CA^{-1}B\right\} .\label{eq:Block1f}
\end{gather}
\end{prop}

\subsection{Joukowski transform\label{subsec:jouk}}

Joukowski (Zhukovsky) transform, \cite[26]{FukSha}, \cite[8.4]{Simon C2A},
is defined by
\begin{equation}
u=z+\frac{1}{z}.\label{eq:jouk1a}
\end{equation}
The inverse to it is of the degree 2, that is there are two solutions
$z$ to the equation (\ref{eq:jouk1a}) for $z$, namely
\begin{equation}
z_{1}=\frac{u-\sqrt{u^{2}-4}}{2},\quad z_{2}=\frac{u+\sqrt{u^{2}-4}}{2}=\frac{1}{z_{1}}.\label{eq:jouk1b}
\end{equation}
The branch of $\sqrt{u^{2}-4}$ in equations (\ref{eq:jouk1b}) is
defined so that $\sqrt{u^{2}-4}=u\sqrt{1-\frac{4}{u^{2}}}$ where
the continuous branch of the square root $\sqrt{1-\frac{4}{u^{2}}}$
outside interval $\left[-2,2\right]$ is chosen by the condition that
for $u=\infty$ we have $\sqrt{1}=1$. For detailed description of
the Riemann surface associated with function $u=z+\frac{1}{z}$ and
its inverse function branches in equations (\ref{eq:jouk1b}) see
\cite[26]{FukSha}.

The Joukowski transform $u=z+\frac{1}{z}$ provides for one-to-one
correspondence between the unit disk $\mathbb{D}=\left\{ z:\left|z\right|<1\right\} $
and set $\mathbb{\hat{C}}\setminus$$\left[-2,2\right]$ where $\mathbb{\hat{C}}=\mathbb{C}\cup\left\{ \infty\right\} $,
that is the set of complex numbers $\mathbb{C}$ extended by adding
to it $\infty$. The first equation in (\ref{eq:jouk1b}) is the inverse
to this one-to-one correspondence. The Joukowski transform $u=z+\frac{1}{z}$
is also one-to-one mapping between (i) the exterior $\left\{ z:\left|z\right|>1\right\} $
of the unit disk and (ii) set $\mathbb{\hat{C}}\setminus$$\left[-2,2\right]$.
The second equation in (\ref{eq:jouk1b}) is the inverse to it. Fig.
\ref{fig:jouk1} shows the polar orthogonal grid of the unit circle
and corresponding under Joukowski transform $u=z+\frac{1}{z}$ the
orthogonal grid on the set $\mathbb{\hat{C}}\setminus$$\left[-2,2\right]$.
\begin{figure}[h]
\begin{centering}
\includegraphics[scale=0.35]{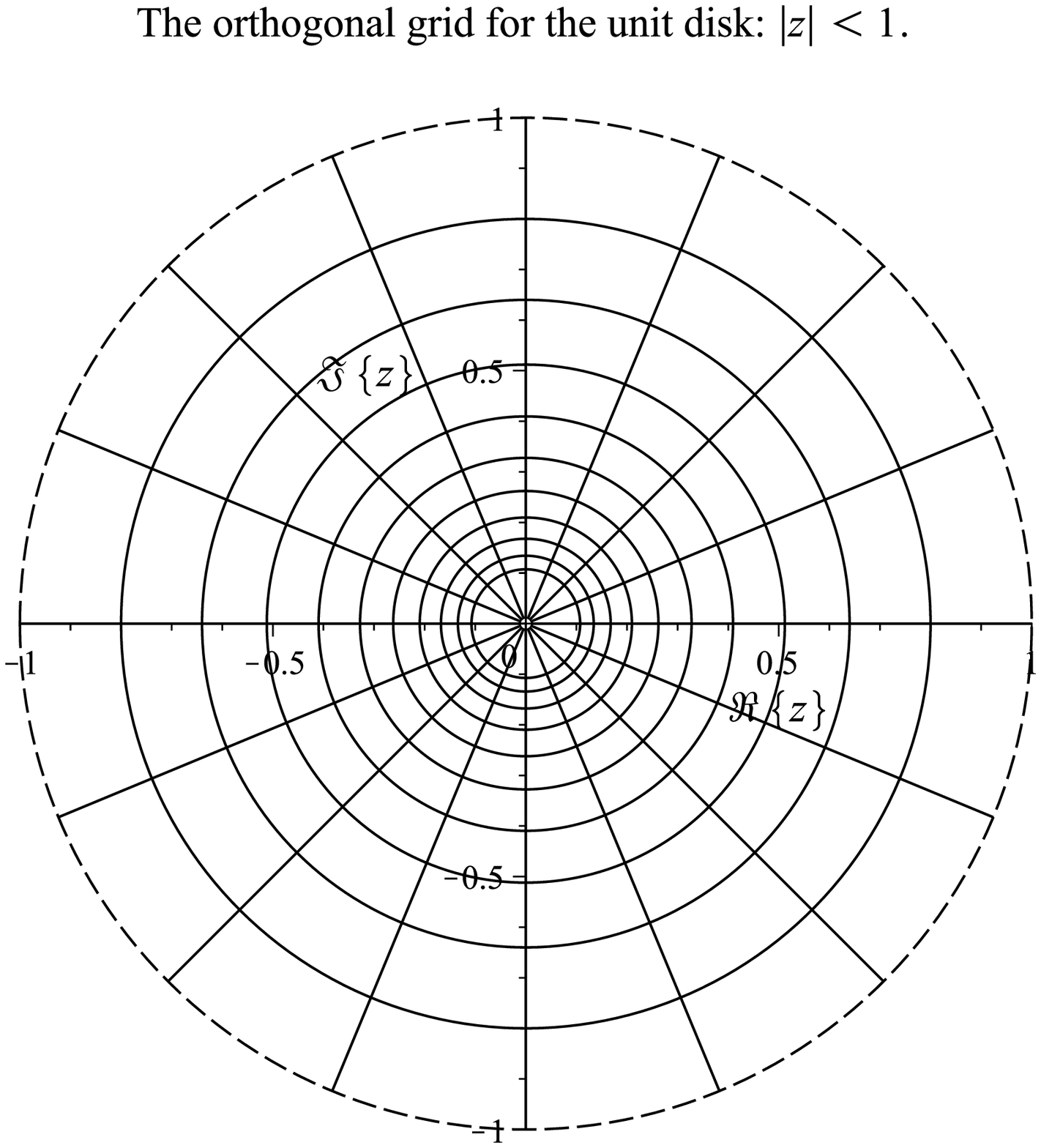}\hspace{0.2cm}\includegraphics[scale=0.43]{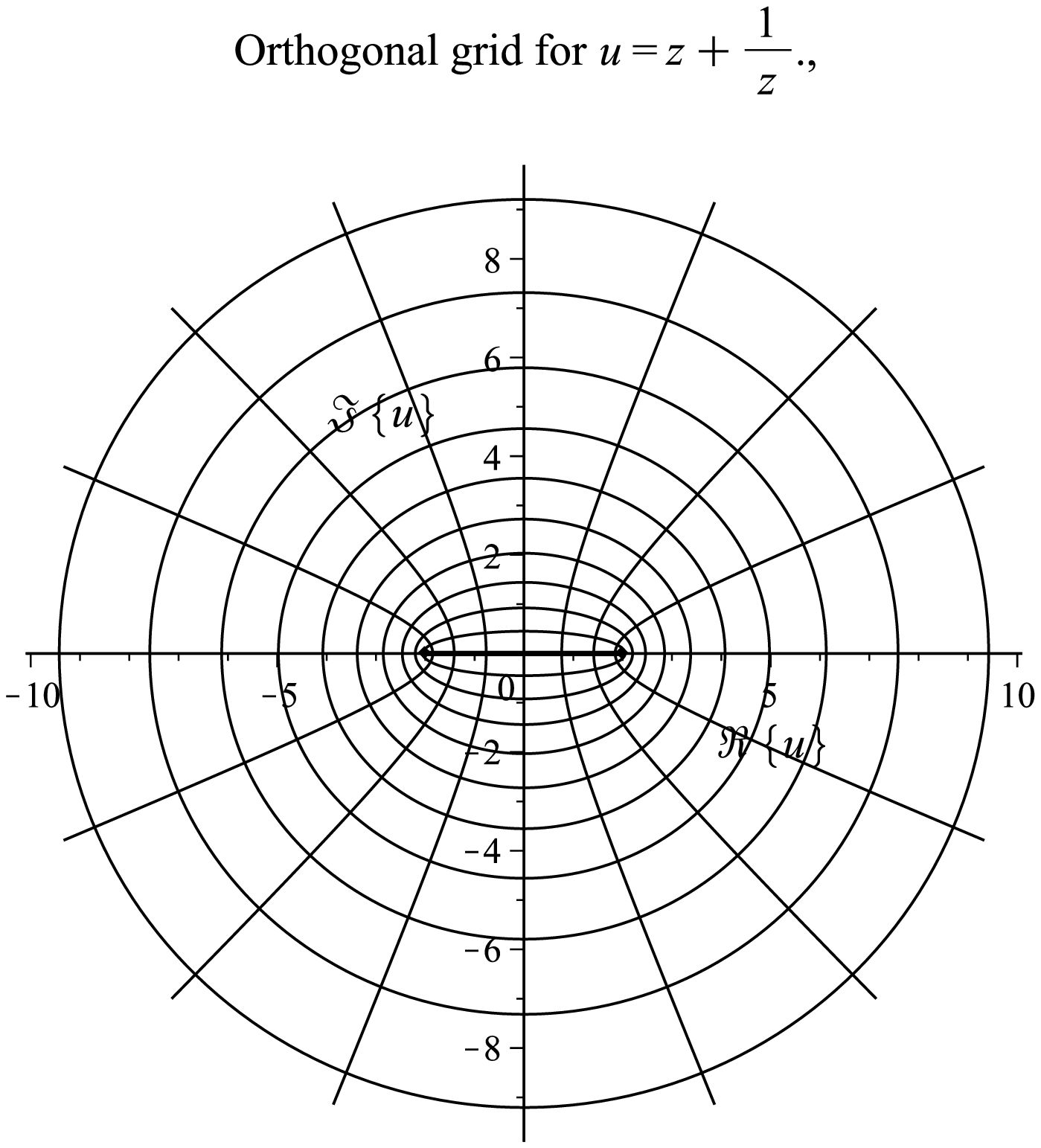}
\par\end{centering}
\centering{}(a) \hspace{7cm}(b)\caption{\label{fig:jouk1} Plot (a) shows the polar-coordinate grid for the
unit disk $\mathbb{D}=\left\{ z:\left|z\right|<1\right\} $. Plot
(b) shows the corresponding under Joukowski transform $u=z+\frac{1}{z}$
the orthogonal grid on the set $\mathbb{\hat{C}}\setminus$$\left[-2,2\right]$.
In particular, counterclockwise oriented circles are mapped onto clockwise
oriented ellipses as described by relations (\ref{eq:jouk1c}).}
\end{figure}
Notice also the image of the circle $\left|z\right|=r<1$ oriented
counterclockwise under Joukowski transform is an ellipse oriented
clockwise, namely \cite[26]{FukSha},
\begin{equation}
z=r\exp\left\{ \mathrm{i}\theta\right\} \rightarrow u=\frac{1}{2}\left(r+\frac{1}{r}\right)\cos\theta+\mathrm{i}\frac{1}{2}\left(r-\frac{1}{r}\right)\sin\theta.\label{eq:jouk1c}
\end{equation}
The circle $\left|z\right|=1$ is mapped onto the doubled segment
$\left[-2,2\right]$ with points $-1$ and $1$ mapped respectively
on $-2$ and $2$. We can think of the upper semi-circle as mapped
on the the lower edge of the cut $\left[-2,2\right]$, and we think
of the lower semi-circle as mapped on the the upper edge of this cut. 

\textbf{Acknowledgment:} This research was supported by AFOSR grant
\# FA9550-19-1-0103 and Northrop Grumman grant \# 2326345.

\end{document}